\documentclass[a4paper,12pt]{article}
\usepackage{tikz}
\usetikzlibrary{positioning,shadows,arrows}
\usepackage{appendix}
\usepackage{times}
\usepackage{color}
\newcommand{\sss}[1]{\subsubsection{#1}}
\usepackage{microtype}%if unwanted, comment out or use option "draft"
\usepackage{url}
\usepackage{amsmath}
\usepackage{graphics}
\usepackage{graphicx}
\usepackage{alltt}
\usepackage{complexity}
\usepackage{algorithmic} 
\usepackage{gensymb}
\usepackage{amssymb}
\usepackage[utf8]{inputenc}
\usepackage[T1]{fontenc}
\usepackage{amsthm}
\usepackage{hyperref}
\newcommand{\emphdex}[1]{\index{#1}\emph{#1}}
\usepackage{latexsym}
\newtheorem{theorem}{Theorem}[section]
\newtheorem{lemma}[theorem]{Lemma}
\newtheorem{corollary}[theorem]{Corollary}
\theoremstyle{definition}
\newtheorem{example}[theorem]{Example}
\newtheorem{definition}[theorem]{Definition}

\newtheorem{proposition}[theorem]{Proposition}
\newtheorem{claim}[theorem]{Claim}
\newtheorem{notation}[theorem]{Notation}
\DeclareMathOperator{\texp}{texp}
\newclass{\NPFP}{NPFP}
\newclass{\APFP}{APFP}
\newclass{\NIFP}{NIFP}
\newclass{\AIFP}{AIFP}
\newclass{\PFP}{PFP}
\newclass{\AnH}{AnH}
\newcommand{\ATIME}[2]{\ensuremath{\Sigma_{#2}\TIME(#1)}}
\newcommand{\coATIME}[2]{\ensuremath{\Pi_{#2}\TIME(#1)}}

\newclass{\LFP}{LFP}
\newclass{\LOG}{L}
\newclass{\IFP}{IFP}
\newclass{\HO}{HO}
\renewcommand{\phi}{\varphi}

\newclass{\MHO}{MHO}
\newclass{\MSO}{MSO}
\newclass{\VO}{VO}
\newclass{\KROM}{KROM}
\newclass{\HORN}{HORN}
\newclass{\SO}{SO}
\newcommand{\hoa}[1]{\ensuremath{\HO^{#1}}}
\newcommand{\N}{\ensuremath{\mathbb N}}
\newcommand{\hob}[2]{\ensuremath{\Sigma^{#1}_{#2}}}
\newcommand{\hod}[2]{\ensuremath{\HO^{#1,#2}}}
\newcommand{\hoc}[3]{\ensuremath{\Sigma^{#1,#3}_{#2}}}

\newcommand{\cohob}[2]{\ensuremath{\Pi^{#1}_{#2}}}

\newcommand{\cohoc}[3]{\ensuremath{\Pi^{#1,#3}_{#2}}}
\newcommand{\mhoa}[1]{\ensuremath{\MHO^{#1}}}
\newcommand{\mhob}[2]{\ensuremath{\M\Sigma^{#1}_{#2}}}
\newcommand{\mhod}[2]{\ensuremath{\MHO^{#1,#2}}}
\newcommand{\mhoc}[3]{\ensuremath{\M\Sigma^{#1,#3}_{#2}}}

\newcommand{\ex}[1]{\ensuremath{\exp_{2}^{#1}(n^{O(1)})}}
\DeclareMathOperator{\bit}{bit}
\DeclareMathOperator{\suc}{succ}

\DeclareMathOperator{\plu}{plus}

\DeclareMathOperator{\ifte}{if}
\DeclareMathOperator{\thent}{then}
\DeclareMathOperator{\elset}{else}
\newcommand{\mc}{\mathcal}
\newcommand{\eh}[1]{$#1$th exponential hierarchy}
\newcommand{\mf}{\mathfrak}
\newcommand{\ol}{\overline}
\newcommand{\olmc}[1]{\overline{\mathcal{#1}}}
\newcommand{\ed}{=_{\mathrm{def}}}
\usepackage{authblk}
\date{}
\title{A Note on Higher Order and Variable Order Logic over Finite
  Models
}
\author{Arthur MILCHIOR}
\affil{{LIAFA, Université Paris 7 - Denis Diderot, France\\
 CNRS UMR 7089, Université Paris Diderot - Paris 7, Case 7014\\
 75205 Paris Cedex 13\\
 {\tt Arthur.Milchior@liafa.univ-paris-diderot.fr}
\\LACL, UPEC, Créteil, France\\
LACL, Département d'Informatique\\
Faculté des Sciences et Technologie\\
61 avenue du Général de Gaulle\\
94010 Créteil Cedex
\footnote{Work done in an internship at
University of Massachusetts Amherst under the direction of
Prof. David Mix Barrington.}}}
\begin{document}
\maketitle

\begin{abstract}
  We show that descriptive complexity's result extends in High Order
  Logic to capture the expressivity of Turing Machine which have a finite
  number of alternation and whose time or space is bounded by a finite
  tower of exponential. Hence we have a logical characterisation of
  \ELEMENTARY. We also consider the expressivity of some fixed
  point operators and of monadic high order logic.

  Finally, we show that Variable Order logic over finite structures, a
  notion introduced by \cite{lauri} contain the Analytical Hierarchy.

  % In this article, we intend to extend results of \cite{lauri} about
  % High-Order(\HO) queries over finite structures. The study of those
  % queries is a natural extension of the first and second order queries
  % often used in descriptive complexity and finite model theory.

  % We characterize space bounds and deterministic time bounds by adding
  % operators to these queries, extending the known results on
  % deterministic, non-deterministic and alternating partial and
  % inflationary fixed point and transitive closure over first and
  % second order logic. We discuss the expressivity of monadic high
  % order, and of formulae over high-order vocabulary. We give a better
  % normal form respecting the number of alternations of higher order
  % quantifiers, and we explain how to create arithmetic predicates
  %   over high-order relations. We give conditional results about
  % the collapsing of different classes. Finally, we give a new
  % equivalent definition of Variable Order logic, and prove that it
  % expresses the set of languages in the Analytical Hierarchy.

\end{abstract}

\tableofcontents

\section{Introduction}
\emph{Descriptive complexity} is a field of \emph{computational
  complexity}. It studies the relation between logical formalisms and
complexity classes. For a given complexity class, what logic do we
need to express languages in this class; for a formula in a given
logic, what is the complexity of checking the truth value of this
formula, over finite structures, as a function of the cardinality of the
structure.

The relation between complexity classes and descriptive classes is
strong, since a lot of well known complexity classes, such as
\AC$^0$, \LOG, \NL, \P, \NP, \PSPACE, \EXP\TIME{} and \EXPSPACE{}, are
exactly equal to some descriptive classes using first and second order
relations, with either syntactic restrictions (Monadic relations, Horn
and Krom formulae) or with operators like `` fixed point'' and
``transitive closure'' (see \cite{imm,libkin}).

These issues in terms of capturing complexity classes are well
understood in agreed upon notation for first and second-order logic,
but beyond that there were open questions and a need for clarity and
standardization of notation. 
The extension of those results to higher order logic began with
\cite{leivant}, and was followed more recently by
\cite{lauri,arity,kolo}. It is also called ``Complex object'' in
database theory \cite{database}.

% But many open questions remained, like
% the power of fixpoints, transitive closure, and monadic
% relations. Problems are encountered while extending some definitions
% into high-order. For example, if a formula has free variables of
% high-order, should we consider that the size of the input is the size
% of the universe, or the size of the code of a high-order relation? In
% particular we can study some special cases where the relation
% takes only polynomial space to describe. Another problem is that some
% syntactic restrictions over second order are not easy to extend to
% high-order.

The article \cite{lauri} also introduced the so-called ``Variable-order logic'',
extending the high-order logic where the order of a quantified
variable is not fixed in the formulae. They stated that it is at least
Turing-hard but did not give an upper bound for the expressivity of this
language.

\paragraph{}
The main contributions of this paper are the following:
\begin{itemize}
\item We give a definition of High-Order logic which is less
  restrictive than the usual one,
\item we prove a normal-form theorem which respects the expressivity of
  the logic,
\item we prove the equality between some subclasses of the High-Order
  logic and some complexity classes below \ELEMENTARY{},
\item we prove that any formula in the analytical hierarchy can be
  written as a formula in Variable Order logic.
\end{itemize}

% The structure of this paper is as follows:

% In Section \ref{def} the languages of high-order logics and some
% operators will be formally defined.  Our definition is less
% restrictive than the usual one (of \cite{lauri} for example), but we
% will then prove a normal form theorem for each of those languages,
% showing that there is no loss of generality. Then in Section \ref{hoq}
% we will study the high-order relations for themselves and in Section
% \ref{arith} we will explain how arithmetic predicates can be coded
% into high-order logic. This will let us give equality relations
% between some high-order queries and some complexity classes in Section
% \ref{equ}. In Section \ref{cond} we will show what some
% equalities between some classes of high-order formulae, or between
% some classes of Turing machines, would imply for equalities among
% other classes. Then in Section \ref{vo} we will introduce the
% so-called variable-order logic and prove that it contains  the
% analytical hierarchy. Finally in Section  \ref{open} we will state
% some open problems.
\section{Definition}\label{def}

\subsection{The core of the language}
Let $r$ be an integer. We will begin by defining the syntax of the
$r$th order logic (\hoa r) and its semantics over finite
structures. First Order (\FO) and Second Order (\SO) are the special
cases \hod{1}{} and \hod{2}{}.

\begin{definition}[Universe]
  A \emphdex{universe} $A$ is the set $[0,n-1]$.
\end{definition}
\begin{definition}[Type]\label{type}
  A \emphdex{type} of order 1 is just the element $\iota$, and a type of
  order $r>1$ is a tuple of types of order at most $r-1$.
\end{definition}
\begin{example}For example, $(\iota,((\iota),\iota),\iota)$ is the
  type of a ternary relation of order 4 whose first and last elements
  are elements of the universe, and whose second one is a binary
  relation of order 3 whose first element is a monadic relation of
  order 2 and the second an element of the universe.
\end{example}
\begin{definition}[Relation]
  A \emphdex{relation} of type $\iota$ is an element of the universe,
  and a relation of type $(t_{1},\dots,t_{n})$, where the $t_{i}$'s
  are types, is a subset of the Cartesian product of the relations of
  type $t_{i}$.

  $\R^{t}_{A}$ is the set of all relations of type $t$ over the
  universe $A$.
\end{definition}
\begin{notation}
  In this article $\mc X^{t}$, where $\mc X$ is any symbol, will
  always be a variable of type $t$. Hence $\mc X^{1}$ is a
  first-order variable or a constant.

  For $t\not= \iota$, $\top^{t}$ (resp.  $\bot^{t}$) is the special case
  of the relation of type $t$ that is always
  true (resp. false). 
\end{notation}

By extension, if $\olmc X=\mc X_{1}^{t_{1}},\dots,\mc X_{n}^{t_{n}}$
is a tuple of variables, then we say that $t=(t_{1},\dots,t_{n})$ is
the type of $\olmc X$. 

\begin{definition}[Vocabulary]
  A \emphdex{vocabulary} $\sigma=\{R_1^{t_1},\dots,R_s^{t_s}\}$ is a set of
  relation symbols.   It is a vocabulary of order
  $r$ if the type of every relation is of order at most $r$.

  We denote by $R_i^{t_i}$ a relation symbol of type $t_i$. A relation
  symbol of type $\iota$ is called a constant. We sometimes omit the
  type-superscripts of variables and relations when this information
  is redundant. 
\end{definition}
For example, in $\forall R^{(\iota,\iota)}. R(x,y)\lor \exists
z. x+z=y$, it is clear that the second occurence of $R$ is also of
type $(\iota,\iota)$, that $x,y$ and $z$ are of type $\iota$ and that
$+$ is a predicate of type $(\iota,\iota,\iota)$.

Our definition is not standard in that a vocabulary may include
relations whose type has order greater than 2.

In \cite{lauri} the types are restricted to what we will call arity
normal form in subsection \ref{sec:arity}. We use the more general
definition of types from \cite{arity}, and we will show in subsection
\ref{sec:arity} that our choice is equivalent to their choice (at
least for the complexity classes that we study).
\begin{definition}[Structure]
\label{structure}
For any type $\sigma = \{R_1^{t_1} \ldots\}$, a \emphdex{$\sigma$-structure }$\mf A = (A,
\mf A(R_1), \ldots )$ is a tuple such that $A$ is a nonempty universe and
each $\mf A(R_i) \in \R^{t_{i}}_{A}$.

% A $\sigma$-structure $\mathfrak A$ over $A$ is a function from the
% relation symbols of $\sigma$, such that for each relation $R^{t}$,
% $\mathfrak A(R^{t})\in\R^{t}_{A}$.

When $\mc X^{t}$ and $ R^{t}$ are a variable and a relation of the same
type, then we write $\mathfrak A'=\mathfrak A[\mathcal X/R]$ to speak
of the $\sigma\cup\{\mathcal X\}$-structure such that $\mathfrak
A'(\mathcal X)=R$ and $\mathfrak A'(\mathcal Y)=\mathfrak A(\mathcal
Y)$ if $\mathcal Y\not=\mathcal X$.

By extension, if $\olmc X^{t}=\mc X_{1}^{t_{1}},\dots,\mc
X_{n}^{t_{n}}$ and $\ol R^{t}= R_{1}^{t_{1}},\dots, R_{n}^{t_{n}}$ are
tuples of variables and of relations of the same type then $\mathfrak
A[\olmc X/\ol R]$ is syntactic sugar for $\mathfrak A[\mc X_{1}/
R_{1}]\dots[\mc X_{n}/R_{n}]$.
\end{definition}

\begin{definition}[Formula]
  A \emphdex{high-order formula} $\phi$ is defined recursively as
  usual, such that if $\psi$ and $\psi'$ are formulae then
  $\psi\land\psi', \psi\lor \psi', \neg\psi, \forall \mathcal X^{t}
  \psi$ and $ \exists \mathcal X^{t} \psi$ are also formulae. Here $t$
  is the type of $\mathcal X$.

  Finally, for types $t=(t_{1},\dots,t_{a})$, $\mathcal X^{t}(\mathcal
  Y^{t_{1}}_1,\dots,\mathcal Y^{t_{a}}_a)$ and $\mathcal
  Y^{t}=\mathcal X^{t}$ are the two kinds of atomic formulae.
  % A sentence is a formula without free variables and a query over
  % $\Sigma$ is a formula where the free variables are in $\Sigma$.
\end{definition}
% \begin{notation}
%   We will write $\mathcal X(\mathcal Y_1,\dots,\mathcal Y_a)$ instead
%   of $\mathcal X^{t}(\mathcal Y^{t_{1}}_1,\dots,\mathcal Y^{t_{a}}_a)$
%   when the type of the relation are known, which will often be the
%   case, either because the relation is in the environment or because
%   it is a quantified variable.
% \end{notation}
\begin{definition}[\HO, \hoa{r}, \hob{r}{j} and \hoc{r}{j}{f}]
  The set \emphdex{$\HO$} contains every formulae with high order
  quantifiers, then \emphdex{$\hoa{r}$} is the subset of $\HO$
  formulae whose quantified variables are of order at most $r$. Hence
  $\hoa0$ is the set of quantifier-free formulae.

  The set \emphdex{\hob{r}{j}}(resp. \emphdex{$\cohob{r}{j}$}) for
  $j>0$ is the class of formulae containing
  \cohob{r}{j-1}(resp. $\hob{r}{j-1}$) and closed by conjunction,
  disjunction and existential (resp. universal) quantification of
  variables of order at most $r$. We have \hob{r}{0}=\hoa{r-1}.

  The \emphdex{normal form} of \hob{r}{j}(resp. $\cohob{r}{j}$), where
  $j>0$, is the set of formulae as in equation \ref{hob} with $\psi\in
  \hoa{r-1}$ in normal form and the types $t_{i,j}$ are of order at
  most $r$ (resp. the same kind of formulae, exchanging $\forall$ and
  $\exists$).

  \begin{equation}\label{hob}
    \exists \mc X_{1,1}^{t_{1,1}}\dots
    \exists \mc X_{1,i_{1}}^{t_{1,i_{1}}}\forall \mc
    X_{2,1}^{t_{2,1}}\dots \forall \mc X_{2,i_{2}}^{t_{2,i_{2}}}\dots Q
    \mc X_{j,1}^{t_{j,1}}\dots Q \mc X_{j,i_{j}}^{t_{j,i_{j}}}\psi
  \end{equation}

  We will prove in Subsubsection \ref{sec:normalform} that any formula
  in \hob{r}{j} is equivalent to a formula in normal form.

  Finally \emphdex{\hod rf} (resp. \emphdex{\hoc rjf}) is the subset
  of \hoa r (resp. \hob rj) without any free variables of order more
  than $f$.  This definition is different from the one of \cite{arity}
  where $f$ denotes the maximum arity of high order relations.
\end{definition}
The classes of formulae for $0\le r\le 2$ and $f=2$ are well studied
since $\hoa{0,2}$ is the set of quantifier-free formulae,
$\hoa{1,2}=\FO$ and $\hoa{2,2}=\SO$.

\begin{definition}[Semantics]
  For $r\ge1$, $\mathfrak A\models\mathcal X^{t,r+1}(\mathcal
  Y^{t_{1},r}_1,\dots,\mathcal Y^{t_{a},r}_a)$ if and only if
  \[(\mathfrak A(\mathcal Y_1^{t_{1},r}),\dots,\mathfrak A(\mathcal
  Y_a^{t_{a},r}))\in\mathfrak A(\mathcal X^{t,r+1}).\]

  $\mathfrak A\models\mathcal X^{t,r}=\mathcal Y^{t,r}$ if $\mathfrak
  A(\mathcal X^{t,r})=\mathfrak A(\mathcal Y^{t,r})$, when the last
  equality is an equality of sets. It is decidable since the sets are
  well-founded.

  Satisfaction for $\psi\land\psi', \psi\lor \psi'$ and $\neg\psi$ are
  defined in the usual way.

  $\mathfrak A\models\forall \mathcal X^{t} \psi$ is true if and
  only if for all $R^{t}\in\R^{t}_{A} $, $\mathfrak
  A[\mathcal X^{t}/R^{t}]\models\psi$. 

  $\mathfrak A\models\exists \mathcal X^{t} \psi$ is true if and only
  if there exists some $R^{t}\in\R^{t}_{A} $, $\mathfrak A[\mathcal
  X^{t}/R^{t}]\models\psi$.
\end{definition}
\subsection{Operators}
In this section, if $L$ is a logic class and $P$ is an operator, then
\emphdex{$L(P)$} is the set  that contains the formulae of $L$, closed
by the operator $P$.

\sss{Transitive closure}
\begin{definition}[Transitive closure]
  % Let $\phi$ be an $\sigma\cup\{\mc X_{1}^{t_{1}},\dots,\mc X_{n}^{t_{n}},\mc
  % Y_{1}^{t_{1}},\dots,\mc Y_{n}^{t_{n}}\}$-formula in $L$, where
  % $X_{i}$ and $Y_{i}$ have the same type.

  Let $\olmc X^{t}=\mc X_1^{t_1},\dots,\mc X_n^{t_n}$ be an $n$-tuple
  and let $\olmc Y^{t},\olmc Z^{t}$ and $\olmc T^{t}$ be three other
  $n$-tuples of the same type and let $\phi$ be a $(\sigma\cup\{\mc
  X_{1},\dots,\mc X_{n},\mc Y_{1},\dots,\mc Y_{n}\})$-formula in $L$.

  Then $(\TC_{\ol{\mc X}\ol{\mc Y}}\phi)(\ol{\mc Z}\ol{\mc T})$ is a
  $\sigma$-formula in $L(\TC)$. The operator \TC{}Xe is called the
  ``Transitive Closure'' operator.

  % $\hod rj(\TC)$ is the set of formulae in $\hod rj$ with transitive
  % closure operator.
\end{definition}
\begin{definition}[Semantics of \TC]
  $\mathfrak A\models (\TC_{\ol{\mc X}^{t}\ol{\mc
      Y}^{t}}\phi)(\ol{\mc Z}^{t}\ol{\mc T}^{t})$ is true if and only
  if $\olmc T=\olmc Z$ or if there exists an $n$-tuple $\olmc M^{t}$
  of type $t$ such that $\mathfrak A[\olmc
  X/\mf A(\olmc Z)][\olmc Y/\mf A(\olmc M)]\models\phi$ and $\mathfrak A\models
  (\TC_{\ol{\mc X}\ol{\mc Y}}\phi)(\ol{\mc
    M}\ol{\mc T})$.\end{definition}

% We may also use the equivalent definition that $\mathfrak A\models
% (\TC_{\ol{\mc X}\ol{\mc Y}}\phi)(\ol{\mc Z}\ol{\mc T})$ is true if and
% only if $\olmc T=\olmc Z$, $\mathfrak A[\olmc X/\olmc Z][\olmc Y/\olmc
% T]\models\phi$ or if there exists an $n$-tuple $\olmc M^{t}$ which is
% equivalent to $\olmc Z$ such that $\mathfrak A\models (\TC_{\ol{\mc
%     X}\ol{\mc Y}}\phi)(\ol{\mc Z}\ol{\mc M})$ and $\mathfrak A\models
% (\TC_{\ol{\mc X}\ol{\mc Y}}\phi)(\ol{\mc M}\ol{\mc T})$.

% In the first definition we guess the path step by step, in the second
% one we guess the middle of the path and continue with a divide and
% conquer method. Both definition are clearly equivalent, 
\begin{example}
  Let the universe be a directed graph and let the vocabulary contain
  only $E$ such that $E(x,y)$ is true if there is an edge from $x$ to
  $y$. Then $\phi=(\TC_{x,y}E(x,y))$ is a relation of type
  $(\iota,\iota)$ such that $\phi(z,t)$ is true if and only if there
  is a path  in the directed graph from $z$ to $t$.
\end{example}

\sss{Fixed Point}
\begin{definition}[Fixed Point]
  Let $\olmc {X}^{t}=\mc X_1^{t_1},\dots,\mc X_n^{t_n}$ be a tuple of
  type $t$ and $\olmc{Y}^{t}$ be another tuple of the same type $t$,
  let $P$ be a variable of type $t=(t_{1},\dots,t_{n})$, and let
  $\phi$ and $\psi$ be some $(\sigma\cup\{P,\mc X_{1},\dots,\mc
  X_{n}\})$-formulae. Then $(\PFP_{\olmc X,P}\phi)(\olmc Y)$,
  $(\IFP_{\olmc X,P}\phi)(\olmc Y)$, $(\NPFP_{\olmc
    X,P}\phi,\psi)(\olmc Y)$, $(\NIFP_{\olmc X,P}\phi,\psi)(\olmc Y)$
  , $(\APFP_{\olmc X,P}\phi,\psi)(\olmc Y)$, $(\AIFP_{\olmc
    X,P}\phi,\psi)(\olmc Y)$ are $(\sigma\cup\{\olmc Y\})$-formulae in
  \emphdex{$L(\PFP)$}, \emphdex{$L(\IFP)$}, \emphdex{$L(\NPFP)$},
  \emphdex{$L(\NIFP)$}, \emphdex{$L(\APFP)$} and \emphdex{$L(\AIFP)$}
  respectively. The letters ``N'' and ``A'' stands for
  ``nondeterministic'' and ``Alternating'' respectively, ``I'' and
  ``P'' for ``Inflationary'' and ``Partial'', and ``FP'' stands for
  ``Fixed Point''.
 % $\hod rj(\PFP)$ is the set of formulae in $\hod rj$ with a partial
 % fixed point operator.

  We restrict the formulae of $\NIFP$ and \NPFP{} such that there are no
  negation applied outside of a non-deterministic fixed-point
  operator.
\end{definition}

\begin{definition}[Semantics of \PFP]
  Let $(\PFP_{\olmc X^{t},P^{t}}\phi)(\olmc Y^{t})$ be a formula. Then we can
  define the relations $(P_{i}^{t})_{i\in N}$ by recursion on $i$.

  For each $\olmc X\in\R^{t}_{A}$, $P_0(\olmc X)$ is false and
  $P_i(\olmc X)$ is true if and only if $\mf A[\olmc X/\ol
  R^{t}][P/P_{i-1}]\models \phi$. Hence the property $P_i$ is true on
  the input $\ol R$ if $\phi$ is true on input $\ol R$ when the
  variable $P$ is replaced by the relation $P_{i-1}$.

  Then, either this process leads to a fixed point, i.e. there exists
  $i$ such that $P_{i}=P_{i+1}$ and then $\mf A
  \models\PFP(\phi_{P,\olmc X})(\olmc Y)$ is true if and only if $\mf A
  \models P_{i}(\olmc Y)$ or the set of
  relation of $(P_i)_{i \in N}$ has a cycle of size strictly greater than 1 and
  then $\mf A\not\models\PFP(\phi_{P,\olmc X})(\olmc Y)$.
\end{definition}

\begin{definition}[Semantics of \IFP]
  Using the notation of the last definition, let $\phi'(\olmc
  X,P)=P(\olmc X)\lor\phi(\olmc X,P)$. Then we can define
  $\IFP(\phi_{P,\olmc X})(\olmc Y)$ as $\PFP(\phi'_{P,\olmc X})(\olmc
  Y)$. Another equivalent way to define it is to define $P_0$ as the
  predicate that is always false, and $P_i(\olmc X)=P_{i-1}(\olmc
  X)\lor \phi(P_{i-1},\olmc X)$.
\end{definition}

We should note that to decide if the desired fixed point for $\phi$ exists we
must run the definition step by step and check whether $P_{i+1} =
P_i$ for $i<|R^t_A|$.  The
definition of \IFP{} makes the operator monotonically increasing so a
fixed point will alway be reached within $\log |R^t_A|$ steps.

\paragraph{}
The nondeterministic fixed points and alternating fixed points are
introduced in \cite{nfp}. We choose not to use their notation
``$FP(A,n)(\phi_{1},\phi_{2},S)(\vec t)$'', but instead to use
$(\APFP_{S^{t},\ol x^{t}}\phi_{1},\phi_{2})(\ol t)$ to be coherent
with the notation for $\PFP$ as defined in \cite{imm}.

\begin{definition}[Semantics of $\NPFP$ and \NIFP]\label{dnpfp}
  Let $\olmc {X}^{t},\olmc{Y}^{t}$ be two vectors of the same type $t$,
  $P^{t}$ be a variable of type $t$, and $\phi_{0}$  and $\phi_{1}$ be
  $(\sigma\cup\{P^{t},\ol X^{t}\})$-formulae.

  We can define the relations $((P_{l})_{l\in\{0,1\}^{*}})$, where $l$ is
  a list of bit. For each $\olmc X\in\R^{t}_{A}$, $P_{\epsilon}(\olmc
  X)$ is false, and by induction for $i\in \{0,1\}$, $P_{il}$ is true
  iff $\mc A[\olmc X/R][P/P_{t}]\models\phi_{i}$.

  Then $\mc A[\olmc Y/\ol R]\models(\NPFP_{\olmc X,P}\phi_{0},\phi_{1})(\olmc
  Y)$ is true if and only if there exists an $l\in\{0,1\}^{*}$ such that $\mc
  A[\olmc Y/\ol R]\models P_{l}(\olmc Y)$ and $P_{0l}=P_{{1l}}=P_{l}$.

  This means that $(\NPFP_{\olmc X,P}\phi_{0},\phi_{1})$ is the union of the
  relations which are fixed points $P$ for both $\phi_{ i}$ and
  such that $P$ is accessible by applying the $\phi_{i}$ a finite
  number of time to $\bot$.

  We could also consider this as a directed graph $G_{\phi_{0},\phi_{1},\mf
    A}$ without self-loop, with a node from every relation $P$ to
  $\phi_{i}(P)$. Then $(\NPFP_{\olmc
    X,P}\phi,\psi)\ed\bigcup_{s}\{P_{s}|P_{0s}=P_{1s}=P_{s}\}$ is
  the union of the leaves reachable from the relation $\bot$. % Then
% $\mf A[\olmc Y/\ol R]\models(\NPFP_{\olmc X,P}\phi,\psi)(\olmc Y)$ is
% true if and only if there exists a list $\xi_{1}\dots\xi_{j}$ where
% $\xi_{i}$ is $\phi$ or $\psi$ such that $\mf A[\olmc Y/\ol
% R]\models(\xi_{1}(\dots(\xi_{r}(\emptyset))))(\olmc Y)$.

  The semantic of $\NIFP$ is to $\NPFP$ what $\IFP$ is to $\PFP$. This
  means that for every relations $P$, $\phi(P)\subseteq P$ and
  $\psi(P)\subseteq P$. It is also possible to define $(\NIFP_{\olmc
    X,P}\phi,\psi)$ as syntactic sugar for $(\NPFP_{\olmc X,P}P(\olmc
  X)\lor\phi(P,\olmc X),P(\olmc X)\lor\psi(P,\olmc X))$. The graph
  defined above is then acyclic.
\end{definition}
%\sss{Alternating fixed point}
\begin{figure}
\begin{tikzpicture}[
    node/.style={rectangle, draw=black,
        text centered, anchor=north},
    level distance=0.5cm, growth parent anchor=south
]
\node  (Node01) [node] {$\emptyset$} [sibling distance=4cm]
child{
  node (Node02) [node] {$P_{0}=\phi(\emptyset)$}
  child{[sibling distance=2.8cm]
    node (Node05) [node] {$P_{01}=\psi(P_{0})$}
    child{
      node (Node03) [node] {$P_{010}=\phi(P_{01})$}
      child{
        node (Node06) {\vdots}
      }
    }
    child{
      node (Node04) [node] {$P_{011}=\psi(P_{01})$}
      child{
        node (Node07) {\vdots}
      }
    }
  }
}
;
\end{tikzpicture}
\caption{A part of an alternating tree, with
  $s(0)=s(1)=1$ and $\psi(P_{11})=P_{11}$}
\end{figure}
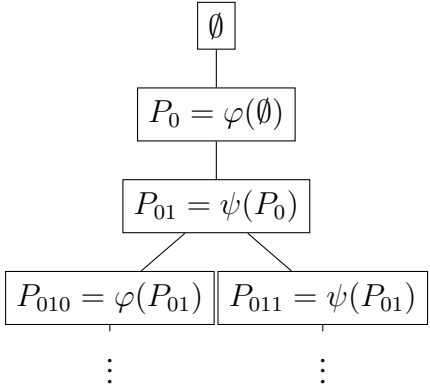
\begin{definition}[Semantics of $\APFP$ and $\AIFP$]
  We use the notations of definition \ref{dnpfp}.

  Let $s$ be a function from strings of bits to bits.  Let $\sigma$ be
  a vocabulary, $\mc A$ be a $\sigma$-structure and $\phi_{0}$ and $\phi_{1}$
  be $\sigma\cup\{P^{t}, \olmc X^{t}\}$-formulae. Then we define the
  tree $T_{\phi_{0},\phi_{1},s,\mf A}$ whose nodes are labelled by relations
  of type $t$. The root is the relation $\bot$, and for $n$, a list of
  bits that indicates a path from the root in the tree, we define the
  label of $n$ as $P_{n}$, as in definition \ref{dnpfp}. If
  $P_{n}=P_{0n}=P_{1n}$ then $n$ is a leaf, else if the depth of $n$
  is even then its children are the nodes with labels $P_{0n}$ and
  $P_{1n}$ that are not equal to $P_{n}$, else its only child is
  $P_{s(n)n}$. We assume that $P_{s(n)n}\not=P_{n}$ else we consider
  that the tree $T_{\phi_{0},\phi_{1},s}$ does not exist.

  A \emphdex{local alternating fixed point} $A_{\phi_{0},\phi_{1},s,\mf A}$ is a
  relation such that $\mf A\models A_{\phi_{0},\phi_{1},s,\mf A}(\olmc Y)$ if
  and only if for every label $l$ that are leaves of an existing tree
  $T_{\phi_{0},\phi_{1},s,\mf A}$ we have $\mc A\models P_{l}(\olmc Y)$. This
  means that a tuple is accepted by the tree if and only if it is
  accepted by every relations of its leaves.

  The \emphdex{alternating fixed point}, $A_{\phi_{0},\phi_{1},\mf A}$ is a
  relation such that $\mf A\models A_{\phi_{0},\phi_{1},\mf A}(\olmc Y)$ if
  and only if there exist an $s$ such that $\mf A\models
  A_{\phi_{0},\phi_{1},s,\mf A}(\olmc Y)$.

  Then $\mf A\models(\APFP_{\olmc X,P}\phi_{0},\phi_{1})(\olmc Y)$ is true if
  and only if $\mf A\models A_{\phi_{0},\phi_{1},\mf A}(\olmc Y)$.

  $\AIFP$ is to $\APFP$ what $\NIFP$ is to $\NPFP$.
\end{definition}

This is almost the definition of \cite{nfp}, except that
$T_{\phi_{0},\phi_{1},s,\mf A}$, $A_{\phi_{0},\phi_{1},\mf A}$ and $A_{\phi_{0},\phi_{1},s,\mf
  A}$ are not named and $s$ is not considered, but having a name for
those values will help the proof of \ref{eqab}. On page 8 they speak
of the ``length of the longest branch'', and it seems that they assume
that the tree is of finite size. They do not seem to explain why this
assumption can be true without loss of generality and without
considering that the tree is instead a graph; it is easy to imagine a
branch which repeats itself an infinite number of times when
$G_{\phi_{0},\phi_{1},\mf A}$ is cyclic. Hence we think it is interesting to
give another definition of alternating fixed point where we can always
give an answer in a finite time .

\begin{definition}[$T_{\phi_{0},\phi_{1},\mf A}$]
We will write $\ol \cup$ for $\cap$ and $\ol \cap$ for $\cup$.

Let $T_{\phi_{0},\phi_{1},\mf A}$ be a tree where each node's label is a pair
with either $\cup$ or $\cap$ as first element and a relation as second
element, and where the root is $(\cup,\bot)$. The children of $(c,P)$ are $(\ol
c,\phi_{0}(P))$ and  $(\ol c,\phi_{1}(P))$ except if $P=\phi_{1}(P)=\phi_{0}(P)$ in
which case this node is a leaf. If in a branch we find two nodes with
the same label $(c,P)$, we remove the second occurrence and its
descendants.

We recursively define the output of the tree as the relation of the
label if the tree is a leaf, else as $c$ applied to the output of its
children. By extension we write $T_{\phi_{0},\phi_{1}}$ instead of its
output relation. It will be clear by the context if we mean the tree
or its output.
% This tree is considered as a circuit over set, we only consider the
% relation of the leaves, for non-leaf nodes we consider the relation
% $c$ of their children, and the output of the tree is the output of
% the We write $n(b,P,\{n_{1},n_{2}\})$ to means that $n$ is a node
% whose label is $(b,P)$ and whose sons are $n_{1}$ and $n_{2}$. $r$
% is recursively defined on $T$ as:
% \begin{eqnarray*}
%   r(n(b,P,\{\})):= P\\
%   r(n(1,P,\{n_{1},n_{2}\})):= r(n_{1})\cap r(n_{2})\\
%   r(n(0,P,\{n_{1},n_{2}\})):= r(n_{1})\cup r(n_{2})\\
%   r(n(b,P,\{n_{1}\})):= r(n_{1})\\
% \end{eqnarray*}
\end{definition}
\begin{proposition}\label{eqab}
  $T_{\phi_{0},\phi_{1},\mf A}=A_{\phi_{0},\phi_{1},\mf A}$
\end{proposition}
\begin{proof}
  Let $\olmc X$ be an tuple. We are going to prove that $\olmc X\in
  T_{\phi_{0},\phi_{1},\mf A}\Leftrightarrow\olmc X\in A_{\phi_{0},\phi_{1},\mf A}$.

  \paragraph{$\Rightarrow$} Let us assume that $\olmc X\in
  A_{\phi_{0},\phi_{1},\mf A}$. Then there exists some function $s$ such that
  $\olmc X\in A_{\phi_{0},\phi_{1},s,\mf A}$. It then suffices to see that on
  every node $n$ of $T_{\phi_{0},\phi_{1},\mf A}$ with label $\cup$ we can
  keep only the child whose number is $s(n)$, and we obtain a tree
  that is a subset of $T_{\phi_{0},\phi_{1},s,\mf A}$. Since there is no
  negation in the tree, if we remove an element of an union we can not
  add any elements in the output of the tree, hence there is no loss
  of generality in doing that. We now have a tree $T'$ whose only
  gates' label are $\cap$.  It is trivial to see that any element
  $\olmc Y$ is in the output of $T'$ if and only if it is in every
  leaf. Since in the construction of $T_{\phi_{0},\phi_{1},\mf A}$ we only
  removed nodes that are copies of nodes higher in the tree, then
  $\olmc Y$ is also in any leaf of $A_{\phi_{0},\phi_{1},s,\mf A}$, hence
  $\olmc X$ is in the output of $T_{\phi_{0},\phi_{1}}$.

  \paragraph{$\Leftarrow$} Let us assume that $\olmc X\in
  T_{\phi_{0},\phi_{1},\mf A}$. We will define a function $s$ such that $\olmc
  X\in A_{\phi_{0},\phi_{1},s,\mf A}$. Note that $s(n)=0$ if $n$ is a string
  of odd size, since then the value of $s$ does not matter in
  $A_{\phi_{0},\phi_{1},s,\mf A}$. Let $n$ be the shortest string such that
  $s(n)$ is not defined. Then $n$ is of even length, hence by
  hypothesis over $T_{\phi_{0},\phi_{1},\mf A}$, $\olmc X$ is in the node $n$.
  Either $n$ is a union node, in which case there is a child $b$ such
  that $\olmc X$ is in $bn$, and we define $s(n)$ as $b$ and for every
  finite string $m$, $s(mn)$ as $0$ since those values do not
  matter. Else $n$ is a leaf. If it is because $P_{n}=P_{0n}=P_{1n}$
  then $s(n)=0$ since this value does not matter. Else it is because
  its children were already seen in this branch, in which case let $m$ be the
  other occurrence of a node with the same relation, $n=pm$ and for
  every $q$ we define $s(qn)$ as $s(qm)$, by hypothesis over $n$ it is
  well defined, since $n$ is the shortest non defined string, and
  $q$'s length is strictly positive.
\end{proof}

When we cut in $T_{\phi_{0},\phi_{1},\mf A}$, it was because the child
was an infinite repetition of itself, and we  can define the
function $s$ in $A_{\phi_{0},\phi_{1},s}$ with the same repetition. It
is then trivial to see that $\olmc X$ is indeed in every leaf of
$T_{\phi_{0},\phi_{1},s,\mf A}$, hence in $A_{\phi_{0},\phi_{1},\mf
  A}$.
\begin{claim}\label{AFP}
  In fact, the same proof would work for a tree bigger than
  $T_{\phi_{0},\phi_{1}}$, choosing to cut later in the branches would not
  remove anything since there are no $\neg$ gates, and would not add
  anything in the output since the later $\cap$ gates would remove the
  eventual new elements of the set.
\end{claim}
This will be useful since it means we will not have to remember the
set of relations seen on a branch, and we only have to count until we
have seen more nodes than the number of relations.

\sss{Operator normal form}\label{opnorm}

It was proved in \cite{imm} that the transitive closure and
deterministic fixed points can be in normal form without loss of
generality. In fact algorithms were given to obtain equivalent formulae
in normal form. Furthermore, \cite{nfp} states that this normal form extends to
alternating fixed points, and to nondeterministic fixed points that
are not under negation (which is impossible by definition).

Those normal forms are $(\TC_{\olmc X,\olmc Y}\phi)(\bot,\top)$ and
$(F_{P,\olmc X}\phi)(\bot)$ where $F$ is a fixed point operator and
$\phi$ a formulae in $\FO$ or $\SO$, it is trivial that this result
extends to high order.

\subsection{Mathematics definitions and notations}

\sss{Mathematics functions}
\begin{definition}[Iterated exponential] %and logarithm
  Using the standard notation for the tetration operator, we define :
  $\exp_i^n(x)=i^{\exp_i^{n-1}(x)}$ and $\exp_i^0(x)=x$. That is
  $\exp_i^n(x)=i^{i^{i^{\dots^{i^{x}}}}}$ with $n$ exponentiations of
  $i$.
  % Then let us write $\log_i^{x*}(a)$ the smallest $n$ such that
  % $\exp_i^n(x)>a$.
We will also write $\texp_a^n(x,r)=a^{r\times\texp_a^{n-1}(x,r)}$
  and $\texp_a^0(x,r)=x$.
\end{definition}

\begin{definition}[Elementary function]
  Let $f$ be a function from $\N$ to $\N$, it is an elementary
  function if $f= \exp_{2}^{O(1)}(n)$, that is, there is a constant
  $c$ such that $f=O(\exp_{2}^{c}(n))$. We denote by $\ELEMENTARY{}$ the
  set of languages decidable in elementary time.
\end{definition}

Finally, we introduce some complexity classes that we will use.
\begin{definition}
  Let $f$ be a function from $\N$ to $\N$ and $i\in\N$.

  Let $\A\TIME(f)$ be the set of languages accepted by an alternating
  Turing machine halting in $O(f(n))$ steps on input of size $n$. The
  restriction with at most $i-1$ alternation between universal and
  existential states, beginning by existential, $\Sigma_{i}\TIME(f)$,
  in particular $\Sigma_{1}\TIME(f)$ is denoted by $\NTIME(f)$. The
  definition of $\TIME(f)$ is similar, but every steps are
  deterministic, and so on for $\A\SPACE(f)$ and $\SPACE(f)$, where
  the limit is not on the number of step but on the number of cells
  used by the machine for the computation.

\end{definition}

We could also define \ELEMENTARY{} using space or bounded alternation
since we have $\TIME(f(n))\subseteq\ATIME{f(n)}{O(1)}\subseteq
\SPACE(f(n))\subseteq\TIME(2^{f(n)})$.

\sss{Syntactic sugar in logic}
\begin{notation}
  Let $Q$ be a quantifier, we will define ``$\oplus_{Q}$''.  We
  write``$\oplus_{\exists}$'' for ``$\land$'' and
  ``$\oplus_{\forall}$'' for ``$\Rightarrow$''.

  When we define a language $L'$ from a language $L$, we will always
  assume that ``1'' is a letter that is not in the alphabet of $L$.
\end{notation}

Some formulae will be used often in this article, hence we are going
to define some syntactic sugar in this subsubsection.

$\texttt{card}_{\le a}(\mc T^{(t_{1},\dots,t_{n})})\ed\forall_{0\le i\le
  a, 1\le j\le n}U_{i,j}^{t_{j}}.[\bigwedge_{0\le i\le a} T(U_{i,1},\dots,U_{i,n})\Rightarrow\bigvee_{0\le
  i<j\le a}\ol U_{i}=\ol U_{j}]$,

$\texttt{card}_{\ge a}(\mc T^{(t_{1},\dots,t_{n})})\ed\exists_{1\le i\le
  a, 1\le j\le n}U_{i,j}^{t_{j}}(\bigwedge_{1\le i\le a} T(U_{i,1},\dots,U_{i,n})\bigwedge_{1\le i<j\le a}\ol U_{i}\ne \ol U_{j})$,

and $\texttt{card}_{a}(\mc T^{p})\ed\texttt{card}_{\ge a}(\mc
T^{p})\land \texttt{card}_{\le a}(\mc T^{p})$.

On ordered set we define $0(x)\ed\neg \exists y (y<x)$,
$\max(x)\ed\neg \exists y (y>x)$, and $1(x)\ed card_{1}(y<x)$ where
$y$ is the free variable of the formula ``\texttt{card}'' to means
that $x$ is 0, 1 or max. We will assume that we can use those constants
without having to explicitly quantify them in the formulae.

Finally, $(Qx. \phi)\psi $ is syntactic sugar for $Qx(\phi\oplus_{Q}\psi)$.

% \sss{Bounded alternating time Turing machines}
% \begin{definition}[Bounded alternating time]
%   Let $C$ be a class of integer functions, and let $a\in \mathbb N$.   Then
-%   $\ATIME{C}{a}$ is the class of the Turing machines using $O(f(n))$

\subsection{Normal form}
In this subsection, we are going to discuss two ways to normalize the
language and see that the definition we choose does not change the
expressivity of the language. Hence we will be able to choose the more
restrictive one to prove theoretical results, and the more expressive
one to express queries.  These results are on the syntax of the
formula, hence they also extend as results for general logic, with
finite or infinite models.

\sss{Types of fixed arity}\label{sec:arity}
We are going to restrain type to a special form and prove that it does
not change the expressivity of the language.

\begin{definition}[Arity relation]
  For each $a,r\ge 1$, we define \emphdex{$A(a,r)$}to be the type
  $(A(a,r-1),\dots,A(a,r-1))$ if $r>1$, with $a$ copy of $A(a,r-1)$
  and we define $A(a,1)$ to be $\iota$. We write $F^{a,r}$ for
  $R^{A(a,r)}$,   the set of relations of type $A(a,r-1)$.

  We say that a formula is in \emphdex{arity normal form}
  (\emphdex{ANF}) if all of its types respect the arity
  definition. Let us define $\Sigma'$ to be the set of formulae in
  ANF.
\end{definition}

\begin{proposition}
  The class of queries of $\Sigma'^{r}_{j}$ is exactly the class of
  queries of $\hob{r}{j}$. Formally for every formula
  $\phi\in\Sigma^{r}_{j}$ we can find an equivalent formula
  $\phi'\in\Sigma'^{r}_{j}$.
\end{proposition}
\begin{proof} The side $\subseteq$ is trivial, since the definition of
  $\Sigma'$ is a restriction of the definition of $\HO$. Indeed ``$a,r$'' as
  a type is defined as ``$\iota$'' if $r=1$ and as
  ``$(a,r-1),{\dots},(a,r-1)$'' where $(a,r-1)$ is considered as a
  type.

  To show $\supseteq$, let $\phi\in\hob{r}{j}$ and define $a$ to be
  the size of the bigger tuple, defined this way: size(1)$\ed$1 and
  size$(t_{1}\dots,t_{a})\ed\texttt{max}(a,\texttt{max}_{1\le i\le
    a}\texttt{size}(t_{i}))$. There are two problems that we need to
  correct. First we need to change the type of every relation such
  that a type of order $r$ contains only type of order $r-1$, and such
  that all relations of those types have the same arity.

  \paragraph{Step normal form}Let us define \emphdex{step normal form}
  (\emphdex{SNF}) to be the formulae that respect the first of these
  properties, that a type of order $r$ contains only types of order
  $r-1$.  We are going to show that each formula is equivalent to a
  formula in SNF. To do this, the encoding in order $j$ of a relation
  $\mc R^{i}$ of order $i$, when $i<j$, will be a relation of order
  $j$ whose type contains only one elements of arity $j-1$, whose type
  contains only one elements, and so on until the one element of order
  $i$ which is of course $\mc R^{i}$.

  Let us define a formula $\texttt{equiv}(\mc S_{i},i,\mc
  S_{j}^{(\dots(\iota)\iota)},j)$ with $j-i$ pair of parenthesis, that
  is true if and only if $S_{j}$ is interpreted  as explained in the
  last paragraph.
  \begin{eqnarray}
  \label{eq:equiv}
  \texttt{equiv}(\mc S_{i},i,\mc S_{j},j)\ed \exists_{i<k<
    j}\mc S_{k}\bigwedge_{i<k\le j}\mc S_{k}(\mc S_{k-1})\land(\forall T\mc S_{k}(T)\Rightarrow T=\mc S_{k-1})
\end{eqnarray}
It now suffices to replace every instance of an atomic proposition
like $\mc X^{r}(\mc Y_{1}^{r_{1}},\dots,\mc Y_{n}^{r_{n}})$ by
$\exists_{1\le k\le n}\mc S_{k}^{r-1}\mc X^{r}(\mc
S_{1}^{r-1},\dots,\mc S_{n}^{r-1})\bigwedge_{1\le k\le
  n}\texttt{equiv}(\mc Y_{i}^{r_{i}},r_{i},\mc S_{i}^{r-1},r-1)$.  An
easy induction shows us that we obtain an equivalent formula, and it
is clear that it is in SNF.

\paragraph{From SNF to arity normal form:}
From now on we will assume that every type in the vocabulary respects
ANF.

We will make sure that every quantified relation is of arity $a$, and
we will do it so that in every relation of arity $b<a$, the last
element will be copied $a-b+1$ times. We need to check that, when
relations are quantified, they respect this property.
% About input relational to encode the elements of the structure such
% that they respect this property.
\begin{eqnarray*}
  \texttt{encode}_{a,b}(\mc X^{a,r})\ed\forall_{1\le i\le a}\mc Y_{i}^{a,r-1}(\mc X(\mc Y_{1},\dots,\mc Y_{a})\Rightarrow \bigwedge_{b<j\le a}\mc Y_{r}=\mc Y_{j} )
\end{eqnarray*}

% Since the input predicate can be a relation with a type that can not
% be described by an arity, there is no way to deal with them in our
% normal form. However, it is easy to create an input relation
% respecting the arity convention from one not respecting it, and hence
% we will not discuss this encoding problem any more. Now we are going
% to take care of quantified variables.

We will assume that $\phi$ is in SNF.  Syntactically we will replace
every occurrence of $Q \mc X^{t,r}\phi$ by
\begin{eqnarray}
  Q \mc X^{a,r}(\texttt{encode}_{a,b}(\mc X) \oplus_{Q} \phi)
\end{eqnarray}
% where $Q$ is a quantifier and $\oplus=$ ``$\land$'' if $Q=\exists$,
% and $\oplus=``\Rightarrow"$ if $Q=\forall$.
To be more precise, we don't really need the quantification of the
$\mc Y_{i}^{a,r-1}$ to be just after the quantification of $\mc
X^{a,r}$. Since the $\mc Y_{i}$ only interacts with $\mc X$, we can
put the quantifiers anywhere after the quantification of $\mc
X$. Hence if $\phi$ is in decreasing normal
order %\footenote{Definition \ref{dnf} in page \pageref{dnf}}
as defined in the next section (the orders of the quantified variables
decrease) we can postpone the quantifications in $\mc Y$ to put them at
the right place of the list. Then we can, without loss of generality,
extract the quantifiers of $\phi$ to put them after the quantification
of the $\mc Y_{i}$. This way if $\phi$ was in decreasing normal form,
it will remain in normal form.

We will replace every atomic formula $\mc X^{t,r}(\mc
Y_{1}^{t_{1,r-1}},\dots,\mc Y_{b}^{t_{n,r-1}})$ by $\mc X^{a,r}(\mc
Y_{1}^{a,r-1},\dots,\olmc Y_{n}^{a,r-1})$ where $\olmc Y^{a,r-1}$
means that the last element is repeated $a-b+1$ times.

It is clear that those formulae are equivalent and in arity normal
form.

% Let $\sigma=\{ R_{1}^{t_{1},r_{1}},\dots, R_{n}^{t_{n},r_{n}}\}$ be
% the vocabulary of $\phi$ and $\{\mc R_{1}^{t_{1},r_{1}},\dots,\mc
% R_{n}^{t_{n},r_{n}}\}$ a $\sigma-$structure, then we will encode our
% formula as $\exists_{1\le i\le n}\phi[R_{i}/ S_{i}]\bigwedge_{1\le
%   i\le n}\texttt{same}(\mc R^{t},\mc S,a,t)$ where we define
% ``same'' like this:

% \begin{alltt}
% same\((\mc{R,S},a,t)\):=match \(t\) with \((t_{1},\dots,t_{n})\)
% \end{alltt}

% \begin{eqnarray}\exists_{1\le i\le n,r_{i}<j<r}(\mc Y_{i}^{j}\bigwedge_{1\le i\le n,r_{i}<j<r}\mc Y_{i}^{j}(\olmc Y_{i}^{j-1}) \nonumber\\
% \bigwedge_{1\le i\le n,r_{i}<j<r}\forall_{1\le k\le a}\mc Z_{k}^{a,j-1}(\mc Y_{i}^{j}(\mc Z_{1}^{j-1},\dots,\mc Z_{a}^{j-1})\Rightarrow (\bigwedge_{1\le k\le a}\mc Y_{i}^{j-1}=\mc Z_{k}^{j-1}))\nonumber\\
% \land \mc X^{a,r}(\mc Y_{1}^{a,r-1},\dots,\olmc Y_{n}^{a,r-1}))
% \end{eqnarray}
\end{proof}

\paragraph{Increasing arity of input structure}
Even if we can not accept an input structure with relations using
``type'', we must at least accept the relation respecting the
``arity'' constraint, which creates a problem when we change the
syntax of our formula. Let say that $R$ is an input structure of order
$r$ and arity $b$, and we want to have a copy $S$ of it of arity
$a$. Then we can state that $S$ is a good copy with
\begin{eqnarray}
  \texttt{copy}_{r}(R^{b,r},S^{a,r})\ed\texttt{encode}_{a,b}(S)\land\forall_{1\le i\le b} \mc Y_{i}^{b,r-1}(R(\mc Y_{1},\dots,\mc Y_{b})\Leftrightarrow\nonumber\\
  \forall_{1\le i\le b}\mc Z_{i}^{a,r-1}((\bigwedge_{{1\le i\le b}}\texttt{copy}_{r-1}(\mc Y_{i},\mc Z_{i}))\Rightarrow S(\mc Z_{1},\dots,\olmc Z_{b}))\\
  \texttt{copy}_{1}(R^{b,r},S^{a,r})\ed R=S
\end{eqnarray}
Let $\sigma=\{R_{1}^{b_{1},r_{1}},\dots,R_{n}^{b_{n},r_{n}}\}$ and let
$\phi$ be a $\sigma$-formula. When we apply the rules of the last
paragraph to extend $\phi$ into an equivalent formula $\phi'$ of arity
$a$ we must in fact transform it into:
\begin{eqnarray}
  \phi''\ed\forall_{1\le i\le n}S_{i}((\bigwedge_{1\le i\le n}\texttt{copy}_{r_{i}}(R_{i},S_{i}))\Rightarrow \phi'[R_{i}/S_{i}])
\end{eqnarray}
\paragraph{Respecting order of quantifiers}
In the next subsubsection we will take care of the order of the
quantifiers, and we will need an algorithm for this normal form, hence
here we must emphasize a few details. Both in the proof of step normal
form and of arity normal form, we did not create any new
quantification of order $r$ or higher, so we respected the global form
of the formula when we consider only $r$th order quantifiers. Now, let
us suppose that we respect the decreasing normal form as in Definition
\ref{dnf}. Then we can see that the new quantifiers of our formula do
not have to be exactly where we put them -- we can postpone them to be
at the good place in the sequence, and postpone the new quantifier-free
part to be with the quantifier-free part of the formula. Hence if the
input is formula in decreasing normal form, the output is also a
formula in decreasing normal form.

Then, as we stated, the normal form equivalent formula is indeed in
$\Sigma'$.

In this article, we will only use formulae in arity normal form,
and this will simplify our proofs, and $\mc X^{a,r}$ will be a
syntactic sugar for $\mc X^{A(a,r)}$.

\sss{Order of the quantifiers}
\label{sec:normalform}
As stated earlier, there is a normal form for $\hob{r}{j}$ and we will
give an algorithm to obtain that normal form. The difference between
this algorithm and the ``folklore'' one as given in
\cite{lauri,Leivant94higherorder} (the latter is about high order in
general and not in finite structures) is that the folklore algorithm
sends $\hoa{r}$ to $\hoa{r}$ but does not respect the number of
alternations. On the other hand, our algorithm sends $\hob{r}{j}$ to
$\hob{r}{j}$. We should note that our algorithm does not respects the
the maximal arity.
% In \cite{arity} it is stated that $\hoc rja\subsetneq\hoc r{c(j)}{a+6}$ where $c(1)=3$ and $c(r)=j+1$ if $j>1$.

\paragraph{Prefix normal form:}
\begin{definition}[Prefix normal form]
  A formula is in \emphdex{prefix normal} form (\emphdex{PNF}) if it begins with a
  sequence of quantifications and ends with a quantifier-free formula.
\end{definition}
\begin{lemma}\label{pnf}
  Every formula in $\hob{r}{j}$ is equivalent to a formula in prefix
  normal form. 
\end{lemma}
We will assume that there are not two variables of the same
  name. Thanks to $\alpha$-conversion this creates no loss of
  generality.
\begin{notation}
  $Q$ will be a meta variable for quantifiers, A for atomic formula,
  $\otimes$ for disjunction or conjunction and $\oplus$ for a polarity
  symbol (+ or -), where applying - to a symbol will give its dual
  while + will not change it. Hence $-\neg\ed\epsilon,
  -\epsilon\ed\neg, -\forall\ed\exists, -\exists\ed\forall,
  -\lor\ed\land,-\lor\ed\land$ and we can even apply a polarity symbol
  to a polarity symbol, $++\ed--\ed+$ and $-+\ed+-\ed-$. On the other hand
  $+\neg\ed\neg, +Q\ed Q$ and $+\otimes\ed\otimes$.
\end{notation}
\begin{proof}(of Lemma \ref{pnf}) We will do a constructive proof, by
  giving an algorithm to transform the formula. We will use three
  auxiliary recursive functions.

\begin{texttt}
  PrefixNormalForm(\(r,\phi\)):=Aux(\(r,\phi,\exists,+\))
\end{texttt}

The result of the Aux function is such that negations are only on
atomic predicates, so it must remember the parity of the number of
$\neg$ it met. This is the information of the last argument. It will
give an output in prefix normal form with as little alternation as
possible and that is why it must know what was the last quantifier of
order $r$. That is what its third argument is for.  Since we want a
normal form for \hob r j, we assume that the formula begins with an
existential quantification of order $r$, and hence we can give an
$\exists$ quantifier as argument.

If Aux meets a quantifier, it will write the very same quantifier and
work inductively on the formula. If it meets a negation it will switch
its polarity and continue inductively. Finally if it finds a
conjunction or disjunction, it will act inductively on both parts to
put them in prefix normal form, and then will combine them with
$\texttt{Aux}'$.

\begin{alltt}
Aux(\(r,\phi,Q,\oplus\)):=match \(\phi\) with
  |\(Q'\mc{X}\sp{a,p}.\psi\rightarrow \)let \(\psi'\)=(if \(p=r\)
     then Aux\((r,\psi,(\oplus{}Q'),\oplus)\)
     else Aux\((r,\psi,Q,\oplus)\))  in \((\oplus{}Q')\mc{X}\psi'\)
  |\(\neg\psi\rightarrow\) Aux\((r,\psi,Q,(-\oplus))\)
  |A\(\rightarrow\oplus\)A   if A is atomic where -A\(\ed\neg\)A.
  |\(\phi\otimes\psi\rightarrow\)let \(\phi'=\)Aux\((r,\phi,Q,\oplus)\) and \(\psi'=\)Aux\((r,\psi,Q,\oplus)\)  in
    Aux'\((r,\phi',\psi',Q,\otimes)\)
\end{alltt}

Of course when we have $Q\mc X$ on the left of the arrow and $Q\mc X$
on the right of the arrow, we assume that both $\mc X$ are of the same
type. This assumption will be true until the end of this proof.

Aux' will take two inputs in prefix normal form, and a parameter to
know if we must consider its conjunction or its disjunction. Then it
will extract from them as many quantifiers of order $r$ of the last
seen polarity as possible. When $\phi$ and $\psi$ begin with
quantifiers of the other polarity of order $r$, we will switch the
polarity we want to extract. Finally when one formula has no more
quantifiers (by the prefix normal form we have by induction, we know
it is then a quantifier free formula) we will extract all predicates
of the other formula using $\texttt{Aux}''$.  Finally we
will link the two quantifier free parts of the formula with the
$\otimes$ relation.
\begin{alltt}
Aux\('\)\((r,\phi,\psi,Q,\otimes)\):= match \(\phi\) with 
  |\(Q'\mc{X}\sp{a,p}.\phi'\rightarrow\)if \(p<r\) or \(Q=Q'\) then \(Q'\mc{X}\).Aux\('\)\((r,\phi',\psi,Q,\otimes)\) else
    match \(\psi\) with
      |\(Q''\mc{Y}.\psi'\rightarrow\)if \(q<r\) or \(Q=Q''\) then \(Q''\mc{Y}\).Aux\('\)\((r,\phi,\psi',Q,\otimes)\)
        else Aux'\((r,\phi,\psi,-Q,\otimes)\)
      |\_\(\rightarrow\)Aux\('\)\('\)\((\phi,\psi,\otimes)\)
  |\_\(\rightarrow\)Aux\('\)\('\)\((\psi,\phi,\otimes)\)
Aux\('\)\('\)\((\phi,\psi,\otimes)\):= match \(\phi\) with 
  |\(Q'\mc{X}\phi'\rightarrow Q'\mc{X}\)Aux\('\)\('\)\((\phi',\psi,\otimes)\)
  |\(_\rightarrow \phi\otimes\psi\)
\end{alltt}
% \end{alltt}
%\begin{alltt}
An easy induction over $\texttt{Aux}''$ and
$\texttt{Aux}'$ shows that the number of alternations in the
output is the larger number of alternations of the two elements of the
input.  Then an induction over Aux shows that its output respects the
same property. It is trivial to see that if the input was a formula of
order $r$, so is the output.
\end{proof}

This algorithm gives a normal form for $\hob{r}{j}$ only, as it does not
promise in general to give the smallest number of alternations. For
example:

\texttt{PrefixNormalForm}$((\forall X\exists Y\phi)\land(\exists Z
\psi))=\exists Z\forall X\exists Z(\phi\land\psi)$

The formula with the smaller number of alternations is $\forall
X\exists Y Z(\phi\land\psi)$. But since this formula begins with a
$\forall$, it is still an $\hob{r}{3}$ formula.

We can easily change the algorithm to obtain a normal form for
$\cohob{r}{j}$, as it is given by
\texttt{Aux}$(r,\phi,\forall,+)$. %It is easy to see that, if both algorithm are correct, the number of alternations for one inp
Finally, if we want an algorithms to obtain the smaller number of
alternation, it suffices to run both algorithm and choose the formula
with the smallest number of alternation.

\paragraph{Decreasing normal form:}
\begin{definition}[Decreasing normal form]\label{dnf}
  An $\hob{r}{j}$ formula, for $r\ge 1$, is in \emphdex{decreasing
    normal form} (\emphdex{DNF}) if it is in the form $\exists \ol{\mc
    X_1^r}\forall \ol{\mc X_2^r}\dots Q\ol{\mc X_j^r}\psi$ where each
  $Q$ is a quantifier and $\psi$ is an $\hoa{r-1}$ formula in
  decreasing normal form.
\end{definition}

\begin{definition}[Normal form]     
  An $\hob{r}{j}$ formula is in normal form (NF) if it is in both
  arity normal form and decreasing normal form, and hence also in
  prefix normal form and step normal form. A formula in $\hoa
  r(\P)$, where $\P$ is an operator, is in normal form if it is in
  operator's normal form and its subformula in $\hoa r$ is also in
  normal form.
\end{definition}

\begin{theorem}
  Every formula $\phi\in\hob{r}{j}$ is equivalent to a formula
  $\phi'\in\hob{r}{j}$ in normal form. 

  And in each group of quantifiers of order $r$, the number of
  quantifiers in $\phi'$ is not greater than the number of quantifiers
  of that order in $\phi$.
\end{theorem}
\begin{proof}
  Let $\phi$ be a formula. Thanks to property \ref{pnf}, we can assume
  it to be in prefix normal form and it will be straightforward that,
  while we transform it, it will remain in prefix normal form.

  The proof will be by induction over the order $r$. It is trivial if
  $r=0$ or $r=1$ because a quantifier free formula and a first order
  formula in prefix normal form are in normal form. Hence we will
  assume that $r>1$ and that the property is true for all $p<r$. Now
  we will prove the property by induction over the number $n$ of
  relations of order $r$. If $n=0$ then it is a formula of order
  $r-1$, and hence the property is true by induction. So we will
  suppose that $n>0$ and that the property is true for every $m<n$. We
  will prove this property by induction over the number $q$ of
  quantifications. It is true if $q=0$ because it is then a
  quantifier-free formula, which is in normal form; we will assume
  $q>1$ and that the property is true for any formula with $r<q$
  quantifications.

  Then $\phi=Q \mc X^{t,i}.\psi$, and by induction over the number of
  quantifiers if $i<r$, or over the number of quantifiers of order $r$
  if $i=r$, there exists a formula $\psi'$ in normal form equivalent
  to $\psi$ . If $i=r$, then $\phi'=Q \mc X^i\psi'$ is in normal form
  and equivalent to $\phi$, and hence the property is true.

  We will now assume that $i<r$. If $\psi'$ contains fewer quantifiers
  of order $r$, then $\phi'=Q \mc X^{t,i}.\psi'$ is a formula, equivalent
  to $\phi$, with fewer quantifiers of order $r$. Hence by the
  induction property over this number we can find an equivalent
  formula in normal form.

  We will then assume that there are at least the same number of
  quantifiers of order $r$ in $\psi'$ as in $\psi$, and since the
  induction hypothesis tells us that there is not more quantification,
  we will assume that the number of quantifications is the same. Since
  $\psi'$ is in normal form and contains a quantifiers of order $r$,
  then $\psi'=Q'\mc Y^{t',r}\xi$ and we can now write $\phi'$ as
  $Q'\mc Y^{t::t',r}Q\mc{X}^{t,i}\xi[\mc Y(\mc Z_1,\dots,\mc Z_a)/\mc
  Y(\mc X,\mc Z_1,\dots,\mc Z_a)]$ where $t::t'$ is the tuple whose
  first element is $t$ and whose other elements are the elements of
  $t'$. Let us assume for now that this formula is equivalent to
  $\phi$. $\psi''=Q\mc{X}^{t,i}\xi[\mc Y(\mc Z_1,\dots,\mc Z_a)/\mc
  Y(\mc X,\mc Z_1,\dots,\mc Z_a)]$ is a formula with fewer quantifiers
  of order $r$ than in $\phi$ and hence it has got a normal form
  $\psi'''$ equivalent to $\psi''$, and then $\phi''=Q'\mc
  Y^{t::t',r}\psi'''$ is a formula in normal form equivalent to $
  \phi$.

  Here with $\mc Y^{t::t'}$ we have lost the normal form of last
  section, so let $\phi''''$ be equivalent to $\phi'''$ and in arity
  normal form. We proved that it is possible, we just need to consider
  the free variables of $\phi$ that are quantified in the entire
  formula as elements of the vocabulary of $\phi$, which is coherent
  with our definition. And, as we explained at the end of last
  subsubsection, since $\phi''$ is in decreasing normal form it will
  remain in this normal form.

  \subparagraph{} Now, it remains to prove that $\phi$ is equivalent
  to $\phi'$ in the last case, which means that $\phi=Q \mc
  X^{t,i}Q'\mc Y^{t',r}\xi$ is equivalent to $\phi'=Q'\mc
  Y^{t::t,r}\psi''$ with $\psi''=Q\mc{X}^{t,i}\xi[\mc Y(\mc
  Z_1,\dots,\mc Z_a)/Y(\mc X,\mc Z_1,\dots,\mc Z_a)]$. There are four
  different cases, for the four possible values of the couple
  $(Q,Q')$. We are going to make a proof for $Q=\forall, Q'=\exists$;
  the three other cases use the same idea.

  Let $\phi=\forall \mc X^{t,i}\exists\mc Y^{t',r}\xi$ and
  $\phi'=\exists\mc Y^{t::t,r}\forall\mc{X}^{t,i}\xi[\mc Y(\mc
  Z_1,\dots,\mc Z_a)/\mc Y(\mc X,\mc Z_1,\dots,\mc Z_a)]$. We are
  going to prove their equivalence, first by proving that the truth of
  the first formula implies the truth of the second one.  Let $\mf A$
  be a structure, and suppose that $\mf A\models\phi$, then for any
  relation $\mc X^{t,i}$ there exists a relation $\mc Y_{\mc
    X}^{t',r}$ such that $\mf A[\mc{X/X}][\mc{Y/Y_{X}}]\models\xi$, so
  let $\mc Y'^{t::t',r}=\{\mc X^{t,i}::\olmc T|\olmc
  T\in\mc{Y_{X}}\}$. Then, for any value of  $\mc X^{t,i}$, $\mc
  Y_{\mc X}(\olmc T)\Leftrightarrow \mc Y'(\mc X,\olmc T)$, and by
  induction over $\xi$, we have $\mf
  A[\mc{X/X}][\mc{Y/Y_{X}}]\models\xi\Leftrightarrow\mf
  A[\mc{X/X}][\mc{Y/X'}]\models\xi[\mc Y(\mc Z_1,\dots,\mc Z_a)/\mc
  Y(\mc X,\mc Z_1,\dots,\mc Z_a)]$.

  Now we will show that the truth of the second statement implies the
  truth of the first. Suppose that $\mf A\models\phi'$. Then there
  exists an $\mc Y^{t::t',r}$ such that for all value of $\mc X^{t,i}$
  we have $\mf A[\mc {X/X}][\mc {Y/Y}]\models\xi[\mc Y(\mc
  Z_1,\dots,\mc Z_a)/\mc Y(\mc X,\mc Z_1,\dots,\mc Z_a)]$. Let $\mc
  X'^{t,i}$ be an arbitrary relation, then let $\mc Y_{\mc X}'=\{\olmc
  T|\mc Y(\mc X',\olmc T)\}$, then $\mc Y_{\mc X}(\olmc
  T)\Leftrightarrow \mc Y'(\mc X,\olmc T)$, and by induction over
  $\xi$ we have $\mf
  A[\mc{X/X}][\mc{Y/Y_{X}}]\models\xi\Leftrightarrow\mf
  A[\mc{X/X}][\mc{Y/X'}]\models\xi[\mc Y(\mc Z_1,\dots,\mc Z_a)/\mc
  Y(\mc X,\mc Z_1,\dots,\mc Z_a)]$.

\end{proof}

\paragraph{Infinite structures}
As stated in the beginning of this subsection, every proof only used
information about the formulae and there is not any use of the
``structure''. Hence this normal form also applies to formulae in
high-order over infinite structures.
\section{High-order queries}\label{hoq}
\subsection{Number of relations}
\begin{definition}Let $r,a>0$ be positive integers. We define $C(r,a)$
  to be the maximum cardinality of a relation of $F^{a,r}$, $N(r,a)$
  to be the number of relations in it, $T(r,a)$ is the number of
  $a$-tuples of relations and $B(r,a)$ is the number of bits necessary to
  describe such a relation. These relations are also defined
  without the ``$a$'', for example $C(r)=C(r,O(1))$.\end{definition}
\begin{lemma}
  We have the following equalities:
  \label{size}
  \begin{itemize}
  \item $C(r,a)=\texp^{r-2}_2(n^a,a)=\exp^{r-2}_2(n^{O(1)})$
  \item$T(r,a)=\texp^{r-1}_2(n^a,a)=\exp^{r-1}_2(n^{O(1)})$
  \item $N(r,a)=2^{\texp^{r-2}_2(n^a,a)}=\exp^{r-1}_2(n^{O(1)})$
  \item $B(r,a)=\texp^{r-2}_{2}(n^{a},a)=\ex {r-2}$
  \end{itemize}
\end{lemma}
This lemma is similar to the one stated in \cite{lauri} but corrects a
minor error there. We need the ``big O'' to be inside of the exponent
and not around it. 
\begin{proof}Indeed, $T(1,a)$ is the size of the Cartesian product of $a$ sets
  of $n$ elements each, so
  $T(1,a)=n^a=\texp^{0}_2(n^a,a)=\exp_2^{0}(n^{O(1)})$.

  By induction, supposing the properties are true up to order
  $r-1\ge1$:
  \begin{itemize}
  \item{}An $a$-ary relation of order $r$ is a subset of the tuples of
    $a$-ary relations of order $r-1$ so
    $C(r,a)=T(r-1,a)=\texp^{r-2}_2(n^a,a)$.
  \item{}Hence the number of $a$-ary relations of order $r$ is the
    number of subsets of the $a$-tuples of $a$-ary relations of order
    $r-1$, so
    $N(r,a)=2^{T(r-1,a)}=2^{\texp^{r-2}_2(n^a,a)}=\exp^{r-1}_2(n^{O(1)})$.
  \item The number of $a$-tuples of $a$-ary relations of order $r$ is
    the size of Cartesian product of $a$ copies of the set of $a$-ary
    relations of order $r$, so
    $T(r,a)=N(r,a)^a=(2^{\texp^{r-2}_2(n^a,a)})^a=2^{\texp^{r-2}_2(n^a,a)\times
      a}=\texp^{r-1}_2(n^a,a)=\exp^{r-1}_2(n^{O(1)})$.
\end{itemize}
The proof for $B(r,a)$ will be the subject of the next subsection.
\end{proof}
\subsection{Encoding relations}\label{code}
In this subsection we will explain how high order relations can be
used and checked in a space-efficient way such that queries of these
relations are also efficient.

Since there are $N(r,a)=2^{{\texp}^{r-2}_2(n^a)}$ relations of order $r$
and arity $a$, we need at least
$\log_2(N(r,a))=\texp^{r-2}_2(n^a)=T(r-1,a)$ bits to encode a relation
$R^{a,r}$ as a string of bits. The last equality is not a surprise,
because all the information one needs to know the relation $R^{a,r}$
explicitly is the set of $a$-tuples of relations of order $r-1$ in
$R^{a,r}$; except for the special case $r=1$, where relations are on
the elements of the universe, but in this case it is well known that
one needs $\lceil\log(n)\rceil$ bits. Since our code will use exactly
this number of bits, it is impossible to find a more space efficient
general encoding.

We will show that this is an exact bound when $r\ge 2$ by creating a
one-to-one encoding function $e$ from a relation of order $r$ and
arity $a$ onto a string of bits of length $\texp^{r-2}_2(n^a,a)$. Let
$b$ be a bit position of $e(R^r)$. As a binary number $b$ is a string
of length $\log_2(\texp^{r-2}_2(n^a,a))=a\texp^{r-3}_2(n^a,a)$, so
inductively, it can be considered as an $a$-tuple of codes of $a$-ary
relations of order $r-1$. The description will then be that the $b$th
bit will be one if and only if this $a$-tuple encodes an element of
$R^r$.

It is clear that this is a one-to-one relation and that this encoding
contains all the relevant information, and thus that the equality of
relations is just an equality of strings of bits. It also gives us a
canonical order over relations, which is the order over the binary
code of the relation.

\subsection{Encoding input}\label{acc}
In subsection \ref{code} we explained how to encode high order
relations in a space efficient way. But it is efficient in the worst
case; in graph theory it would be equivalent to the matrix encoding.
But, as in graph theory, it can also be interesting to consider other
codes for the input, especially for non-dense relations.

An example of a possible code would be a circuit such that the leaves
are elements of the universe, the nodes of height $2n$ are a relations
of order $n$ and the nodes of height $2n+1$ are $a$-tuples of
relations of order $n$. There is an edge from an $a$-tuple into a
relation if this tuple is an element of the relation, and the
$a$-tuples are of in-degree $a$, where there is an order on the edges,
the $a$ predecessors being of course the $a$ elements of the tuple.

Since many different encodings could be imagined, depending on the
assumptions about the problem one wants to solve, we are going to
speak of a more general property.

\begin{definition}[acceptable code]
  An encoding of a $\sigma$-structure $\mathfrak A$ is said to be
  acceptable if for every relation $\mc R^{a,r}$ and $a$-tuple $\olmc
  S^{a,r-1}$ the property $\mc R(\olmc S)$ is decidable in time
  polynomial in the size of the description of $\mc R$ and $\olmc S$.
\end{definition}

% \sss{Reasonable input}
\begin{definition}[reasonable input]\label{reasonable}
  A set of input is reasonable for a given code if the size of the
  code of the structures of this classes is bounded by a polynomial in
  the size of the structure.
\end{definition}
By cardinality, it is clear that the class of every structure of order
at least 3 can not be ``reasonable''.

\begin{claim}
  The circuit encoding and the encoding of section \ref{code} are
  acceptable, and they are reasonable for inputs of order 2.
\end{claim}

\subsection{Reducing the order of the input}
Since we mostly want to study the formulae in \hod r2 we are going to
show how to reduce the order of the structures. That is, for a formula
in \hod {r+1}{f+1} for $f\ge2$, how to obtain an equivalent formula in
\hod {r}{f} over an equivalent vocabulary of order $f$, for a precise
definition of ``equivalent".

We will define the function $F:\hoc {r+1}j{f+1} \rightarrow \hoc rjf$,
and the function $V$ from vocabularies $\sigma$ of order $f+1\ge3$ and
$\sigma$-formulae into vocabularies of order $f$, such that if $\phi$
is an $\sigma$-formula then $F(\phi)$ is a $V(\sigma)$-formula and the
function $S$ is from $\sigma$-structures into $V(\sigma)$-structures such
that $\mf A\models \phi\Leftrightarrow V(\mf A)\models F(\phi)$.

The encoding will be such that $|V(\mf A)|=O(2^{|\mf A|^{O(1)}})$. We
consider that this size is acceptable since an $f$th order relation is
encoded with $\ex{f-2}$ bits, and after we apply those
functions we will have a structure with relations of order up to
$f-1$.  Hence the new size of the encoding of the input will be
$\exp_{2}^{f-3}(2^{n^{O(1)}})=\ex{f-2}$ bits.

Let $\mc R^{a,r}$ be a symbol of order $r$ and arity $a$. We define
$V(\mc R)$ as a symbol of order $max\{1,r-1\}$ and arity $a$.

Let $\sigma=\{<^{2,2},\mc R_{1}^{r_{1},a_{1}},\dots,\mc
R_{n}^{r_{n},a_{n}}\}$, $\phi$ a $\sigma$-formula, let $a'$ be the
highest arity of a quantified variable of $\phi$, let
$a=\max(a_{1},\dots,a_{m},a',1)\footnote{$a$ depends on $a'$ which
  explains why $V$ takes $\phi$ as an input.}$, then
$V(\sigma,\phi)=\{<,n,T_{1}^{2,2},\dots,T_{a}^{a+1,1},V(\mc
R_{1}),\dots, V(\mc R_{n})\}$. If $\mf A$ is a $\sigma$-structure of
cardinality $n$ where $<$ is interpreted as a total order over the
universe\footnote{There is no loss of generality since in high-order
  we can always create a linear order.}, then $S(\mf A)$ contains exactly
$2^{n^{a}}$ elements where the first $n$ elements represent the $n$
elements of $\mf A$, and the first $2^{n^{b}}$ elements, with $b\le
a$, represent the second-order $b$-ary relations, the exact
representation being the same as in subsection \ref{code}. The $c_{i}$
will be the same constants, $n$ will be the $n$th element and represents
the size of the input of the former universe,
$T_{i}^{2,i+1}(x_{0},x_{1}\dots,x_{i})$ will be true if $x_{0}$
represents a second-order $i$-ary relation $R_{x_{0}}^{2,i}$ and the
$(x_{j})_{1\le j\le i}$ represent elements of the universe, and if
$R_{x_{0}}(x_{1},\dots,x_{i})$. 

This means that the same elements of the structures may represent both
a first-order element and second order $b$-ary elements for any
$b$. The exact meaning is known only when the variable is queried.
The former $r_{i}$th order $a_{i}$-ary relations now become
$r_{i}-1$th order $a_{i}$-ary relations, the only other change is that
when $r_{i}-1=2$, we assume that the relation does not accept any
first-order element which is not the representation of a former second
order $a$-ary relation.

\paragraph{} We must define what it means for a relation $R^{q,r}$ to
be a correct encoding of a $b$-ary relation of order $q$. We will do
it with $\text{acc}(r^{q,b})$ which means that it contains no first
order elements that are not encodings of second order elements.
% $\text{acc}(V_{1},\dots,V_{n})\ed\bigwedge_{1\le i\le
%   n}\text{acc}(V_{i})$,
\begin{eqnarray}
  \text{acc}(\mc X^{q,b})\ed \forall_{1\le i\le b}\mc Y^{q-1,b}_{i}(\mc X(\mc Y_{1},\dots,\mc Y_{n})\Rightarrow \bigwedge_{1\le i\le b}\text{acc}(\mc Y_{i}))\\
  \text{acc}(\mc X^{1,b})\ed \mc X<2^{n^{b}}
\end{eqnarray}
Here $n$ is a constant of the new vocabularies which represents the size of
the former universe, we could either  add $2^{n^{b}}$ to the
input, or define $x<{}2^{n^{b}}$ as $\forall x_{0},\dots,x_{b}
(T_{b+1}(x,x_{0},\dots,x_{b})\Rightarrow x_{0}=0)$ if $b<a$, else as
$\top\ed \forall x(x=x)$ and of course $x=0$ as $\neg \exists y.(y<x)$.

\begin{lemma}\label{lem:lower}
  For any high-order relation $\mc R^{r,v}$ we have
  $\text{acc}(V(\mc R))$. If $\text{acc}(\mc R)$ is true then
  there is some $\mc S$ such that $V(\mc S)=\mc R$.
\end{lemma}
\begin{proof}
  The first part is by construction of $V$, and the second one is a
  trivial induction over the order.
\end{proof}

Now we need to define $F$, and we will do it recursively. We assume
without loss of generality that there is no $\forall$ or $\land$. The algorithm is in
table \ref{f}. In this algorithm, when $\olmc V$ is a tuple of
variables we denote $\olmc X=\mc X_{1}^{r_{1},a_{1}},\dots,\mc
X_{r}^{r_{r},a_{r}}$, $V(\olmc{X'})=\mc V(X_{1}),\dots,V(\mc X_{r})$
and if $P$ is a variable whose type is equivalent to $\olmc X$ then
$V(P)$'s type is equivalent to $V(\olmc{X})$.
\begin{table*}\label{f}
\begin{tabular}{|l|l|}
  \hline
  $\phi$&$F(\phi)$\\
  \hline
  $R=S$&$R=S$\\
  $R^{2,b}(x_{1},\dots,x_{b})$&$T_{b}(R,x_{1},\dots,x_{b})$\\
  $\mc R^{p,b}(\mc X_{1},\dots,\mc X_{b})$&$\mc R^{p-1,b}(\mc X_{1},\dots,\mc X_{b})$\\
  $\phi\lor\psi$& $F(\psi)\lor{}F(\psi)$\\
  $\neg\phi$&$\neg F(\phi)$\\
  $\exists{}x.\psi$&$\exists{}x.(x<{}n\land\psi)$\\
  $\exists{}X^{2,b}.\psi$&$\exists{}X^{1}.(x<{}2^{n^{b}}\land\psi)$\\
  $\exists{}\mc X^{p,b}.\psi$&$\exists{}\mc X^{p-1,b}.\psi$\\
  $(\TC_{\ol{\mc X}\ol{\mc Y}}\psi)(\ol{\mc Z}\ol{\mc T})$&$(\TC_{V(\ol{\mc X}\ol{\mc Y})}(\text{acc}(V(\olmc{XY}))\land F(\psi)))(V(\ol{\mc Z\mc T}))$\\
  $(\PFP_{P,\olmc X}\psi)(\olmc Y)$&$(\PFP_{V(P,\olmc{X})}(\text{acc}(V(\olmc X))\land F(\psi)))(V(\olmc{Y'}))$\\
  $(\IFP_{P,\olmc X}\psi)(\olmc Y)$&$(\IFP_{V(P,\olmc{X})}(\text{acc}(V(\olmc X))\land F(\psi)))(V(\olmc{Y'}))$\\
  \hline
\end{tabular}
\caption{F}
\end{table*}
\begin{theorem}
  If $f\ge2, r\ge2$ (resp. $r=1$) and $\phi\in\hoc {r+1}j{f+1}$ then
  $F(\phi)\in\hoc {r}j{f}$ (resp $F(\phi)\in\hod rf$).  For any
  vocabulary $\sigma$, $\sigma$-structure $\mf A$ and $\sigma$-formula
  $\phi$, $\mf A\models\phi\Leftrightarrow V(\mf A,\phi)\models
  F(\phi)$.
\end{theorem}
\begin{proof}
  For the first statement, as we can see, the only new quantifiers are
  of order lower than $r-1$ (resp. of order 1), hence the number of
  alternations of the $r$th order quantification in $F(\phi)$ is the same
  as the number of $r+1$th order quantifiers in $\phi$.
\paragraph{}
For the second statement, we do the proof by induction over
$\phi$. For the atomic formulae it is by construction, and for the
negation, conjunction and disjunction it is trivial.

\paragraph{}So assume that $\phi=\exists x.\psi$, and let us prove $\Rightarrow$. If $\mf
A\models\phi$ then there is some $i<n$ such that $\mf
A[x/i]\models\psi$, by induction $V(\mf A[x/i])\models F(\psi)$ and
since $i<n$, then $F(\phi)$ is true.

For $\Leftarrow$, if $V(\mf A)\models F(\phi)$ then there is some $i$
such that $V(\mf A)[x/i]\models x<c\land F(\psi)$, of course then
$i<n$, hence $V(\mf A)[x/i]=V(\mf A[x/i])$ and by induction $V(\mf
A[x/i])\models F(\psi)\Leftrightarrow \mf A[x/i]\models \psi$, hence
$\mf A\models\phi\Leftarrow V(\mf A)\models F(\phi)$.
 
\paragraph{The case $\phi=(\TC_{\ol{\mc X}\ol{\mc Y}}\psi)(\ol{\mc
    Z}\ol{\mc T})$} 

Let us prove $\Rightarrow$ by induction over the number $s$ of steps
of the transitive closure. If $s=0$ it is trivial, let us suppose that
$s>1$ and it is true for $s-1$. Then there exists $\olmc M$ equivalent
to $\olmc Z$ such that $\mathfrak A[\olmc X/\olmc Z][\olmc Y/\olmc
M]\models\psi$ and $\mathfrak A\models (\TC_{\ol{\mc X}\ol{\mc
    Y}})(\ol{\mc M}\ol{\mc T})$ and then by the induction hypothesis
over $\phi$, we have $V(\mathfrak A[\olmc{X}/\olmc {Z}][\olmc
{Y}/\olmc{M}])\models F(\psi)$ hence $V(\mathfrak A)[\olmc{X'}/\olmc
{Z'}][\olmc {Y'}/\olmc{M'}]\models F(\psi)$ and by lemma
\ref{lem:lower}, $\text{acc}(V(\olmc M))$ and by the induction
hypothesis over $s$, $V(\mathfrak A)\models (\TC_{\ol{\mc X'}\ol{\mc
    Y'}}F(\psi))(\ol{\mc M'}\ol{\mc T'})$.

Now, let us prove $\Leftarrow$, it is also an induction over the number
of steps $s$ that close $V(\mf A)\models(\TC_{V(\ol{\mc
    X}\ol{\mc Y})}(\text{acc}(V(\olmc{XY}))\land
F(\psi)))(V(\ol{\mc Z\mc T}))$. If $s=0$ then it is trivial,
else there exists some $\olmc M$ equivalent to $\olmc X$ such that
$V(\mf A)[\olmc{X}/V(\olmc Z)][\olmc
Y/\olmc{M}]\models\text{acc}(\olmc{XY})\land F(\psi)$ hence by lemma
\ref{lem:lower} there is some $\olmc M'$ such that $V(\olmc M')=\olmc
M$ and by the induction hypothesis over $\phi$ we have  $V(\mf
A[\olmc{X}/\olmc Z][\olmc Y/\olmc{M'}])\models\psi$, which ends the
proof.

The proofs for the fixed points are similar, with induction on the size
of the fixed point.
\end{proof}
\begin{claim}\label{highstructure}
  In this article we will give results for formulae over structures of
  order $2$. In general, if the input structure is of order $p-1$ and
  hence contains at least one relation of order $p$ and no relation of
  higher order, the time and space bound will decrease, by $p-2$
  applications of the logarithm over the bound. In particular, a
  corollary will be that queries in $\hod{r}{r}$ are computable in
  polynomial time, as proven in \cite{complex}.
\end{claim}
% It suffices to see that the real size of the input of the problem
% will be $\exp_{2}^{p-2}(n^{O(1)})$ when $n$ is the size of the
% structures, and that all the classes we will use are closed under
% polynomials.

\section{Arithmetic predicates}\label{arith}
\subsection{Predicates over relations}
In finite model theory, the arithmetic predicates are important,
especially in first order, where even partial fixed points can not
express the parity of the size of the universe without an order
relation. In next sections we will often use either bit predicates or
addition over high order relations, so in this section we will first
explain how to obtain those relations.

As it is already known, a linear order can be specified by a second
order binary relation, hence, contrary to what happens in the
first-order case, we will not make any statement about the
existence or the absence of an order relation in the vocabulary.

We intend to show that the usual predicates that we may ask over first
order, bit, plus, times, <, are redundant in high-order; all of these
predicates can be defined thanks to a first-order total order.

We will speak of some arithmetic operations both over predicates and
over tuples of predicates, as both will be useful in this article. To
distinguish them, we adopt the convention that ``predicate$^{a,r}$''
refers to a predicate over relations and ``predicate$_{a,r}$'' refers
to a predicate over tuples of relations.

\begin{notation}
  In this section, ``$\olmc P^{a,r}$'' will always be an $a$-tuple of
  relations of arity $a$ and order $r$, $\mc P_{1}^{a,r},\dots,\mc
  P^{a,r}_{n}$.
\end{notation}

\begin{claim}[arity of predicate]\label{mon-pred}
  As we will see, to define a predicate over relations of arity $a$,
  quantification is over variables of arity $a$, and hence there is no
  increase of arity of the formula because of the arithmetic
  predicate. In particular, those predicates over monadic relations
  are monadic formulae.
\end{claim}

\sss{Equality predicate}
If there is a binary first-order equality predicate, then every other
equality predicate can be defined in the logic. Define $=^{a,r}$ to be
the equality predicate over relation of order $r$ and arity $a$, and
then we can define it recursively as: $%\begin{eqnarray*}
\mathcal{X}^{a,r}=^{a,r}\mathcal{Y}^{a,r}=_{\mathrm
  {def}}\forall\mathcal{\ol P}^{a,r-1}(\mathcal{X}(\mathcal
{\ol P})\Leftrightarrow\mathcal{Y}(\mathcal{\ol P}))
$%\end{eqnarray*}
. And of course $\olmc P^{a,r}=_{{a,r}}\olmc Q^{a,r}\ed
\bigwedge_{0\le i<a}\mc P_{i}=^{a,r}\mc Q_{i}$.

\sss{Order relation}\label{order}
Suppose that we have an order relation on first-order variables,
$x<y$. Then we can recursively encode a formula
$\mathcal{X}^{a,r}<^{a,r}\mathcal{Y}^{a,r}$ over relations of arity
$a$ and order $r$ considered as binary numbers.

$%\begin{eqnarray*}
\mathcal{X}^{a,r}<^{a,r}\mathcal{Y}^{a,r}=_{\mathrm {def}}\exists
\olmc P^{a,r-1}.(\mathcal{Y}(\mathcal {\ol P})\land
\neg\mathcal{X}(\mathcal {\ol P})\land\forall\mathcal{\ol
  Q}^{a,r-1}(\mathcal{\ol P}<_{a,r-1}\mathcal {\ol
  Q}\Rightarrow (\mathcal{Y}(\mathcal {\ol
  P})\Leftrightarrow\mathcal{X}(\mathcal{\ol Q})))).
$%\end{eqnarray*}

Here $<_{a,r}$ is a relation over $a$-tuples of relations of order $r$
defined as: $%\begin{eqnarray*}
\mathcal{\ol X}^{a,r}<_{a,r}\mathcal{\ol
  Y}^{a,r}\ed\bigvee_{1\le i\le a}(\mathcal X_i<^{a,r}\mathcal
Y_i\bigwedge_{1\le j<i}(\mathcal X_i=^{a,r}\mathcal Y_i)).
$%\end{eqnarray*}

\sss{Bit predicate}\label{bit}
It is usual in descriptive complexity to use a ``bit'' relation,
taking two first order variables $x$ and $y$, such that $\bit(x,y)$ is
true if and only if the $y$th bit of the binary expression of $x$ is
1.

For high order it is easier; since a relation $R^{a,r}$ is equivalent
to a string of $T(r-1,a)$ bits, we can write the $y$ as $a$ relations
of order $i$, and then
\\$\bit(R^{a,r},S_1^{a,r-1},\dots,S_a^{a,r-1})\ed
R^{a,r}(S_1^{a,r-1},\dots,S_a^{a,r-1})$.

\sss{Addition}

The addition of relations is defined as addition over the
corresponding strings of bits.

\begin{eqnarray}
  \phi_{\mathrm{carry}}(\mc X^{a,r},\mc Y^{a,r},\olmc I^{a,r-1})\ed\exists\olmc T^{a,r-1}( \olmc T<_{a,r-1}\olmc I\land\mc X(\olmc T)\land\nonumber  \\\mc Y(\olmc T)\land
  \forall \olmc U^{a,r-1}((\olmc I<_{a,r-1}\olmc U<_{a,r-1}\olmc T)\Rightarrow (\ol X(\olmc U)\lor \ol Y(\olmc U))))  \\
  \plu^{a,r}(\mc X^{a,r},\mc Y^{a,r},\mc Z^{a,r})\ed\forall \olmc{I}^{a,r-1}(\nonumber  \\\mc Z(\olmc I)\Leftrightarrow \mc X(\olmc I)\oplus\mc Y(\olmc I)\oplus\phi_\mathrm{carry}(\mc X,\mc Y,\olmc I))
\end{eqnarray} 

Here $A\oplus B$ is syntactic sugar for $A\Leftrightarrow \neg B$, and
$\phi_{\mathrm{carry}}(\mc X^{a,r},\mc Y^{a,r},\olmc I^{a,r-1})$ is
true if there is a carry propagated in position $\olmc I$ in the
addition of $\mc X$ and $\mc Y$.

% It is interesting to see that an order relation over first order
% variables is enough to obtain addition over high-order relations.

\sss{Addition + Multiplication}
In first-order, it is well-known that addition + multiplication
$\equiv$ bit, and the proof does not specify that the predicate must
be over a first-order object, so the very same proof works for higher
order logic.

Hence, addition + multiplication over first-order elements is
equivalent to the bit predicate over first-order elements, which
extends over higher-order relations as seen in subsubsection
\ref{bit}, which is then equivalent to addition + multiplication over
higher order relations.

\subsection{Addition over tuples}
We will also need to add tuples of elements, and in this subsection we
will show how to do it. Let us define $p=T(r,a)=\texp_2^{r-2}(n^a,a)$.
\paragraph{Overflow:}
We will define $\plu_{a,r}$ over $a$-tuples of relations of arity $a$
and order $r$. First, let $C^{a,r}(\mc X^{a,r},\mc Y^{a,r})$ be a
predicate indicating that the addition of $\mc X$ and $\mc Y$
overflows $(\mc{X+Y}\ge p)$.
\begin{eqnarray*}
  C^{a,r}(\mc X^{a,r},\mc Y^{a,r})\ed\neg\exists\mc Z^{a,r}.\plu^{a,r}(\mc{X,Y,Z})
\end{eqnarray*}
This just means that there is no value $\mc Z$ such that $\mc{X+Y=Z}$,
so $\mc {Z}\ge\texp_2^{r-2}(n^a,a)$.
\paragraph{Addition modulo $p$}
Now we also need to speak of addition modulo $p$, but using only
numbers strictly smaller than $p$.  If the addition does not overflow,
it suffices to test the addition. If it overflows, we can
existentially quantify $d, e,f,g$ and $h$ such that:

\begin{tabular}{ll}
  $d+\mc X=p-1$&$ d=p-1-\mc X$\\
  $d+1=e $&$ e=p-\mc X $ \\
  $f+\mc Y=p-1$&$ f=p-1-\mc Y$\\
%  $f+1=g $&$ g=p-\mc Y $ \\
  $e+g=h $&$ h=2p-\mc X-\mc Y -1$ \\
  $i+h=p-1$&$ i=p-1-(2p-\mc X-\mc Y-1)=\mc X+\mc Y-p$\\
%  $i+1=\mc Z$&$\mc Z=\mc {X+Y}-p$
\end{tabular}

We can then see that if $\mc{X+Y}\ge p$ then each variable has exactly
one possible value which is less then $p$. It is trivial for $d$ and
$f$, and for $e$ it is enough to see that, since $\mc {X+Y}\ge p$ and $\mc
{X,Y}<p$ then $\mc{X,Y}>0$, so $p-\mc {\{X,Y\}}<p$; $g=2p-\mc X-\mc
Y-1\le2p-p-1=p-1$ since $\mc{X+Y}\ge p$, and a fortiori
$h=g+1<p-1+1=p$.

The exact equation of plus modulo ($\plu_m$) is then:
\\
$%\begin{eqnarray}
\mathcal\plu_m^{a,r}(\mc{X}^{a,r},\mathcal{Y}^{a,r},\mc Z^{a,r})\ed\mc{X+Y=Z}\lor\nonumber\\
\exists d,e,f,g.  d+\mc X=(p-1)\land d+1=e\land f+\mc Y=(p-1)\land
e+f=g\land \mc Z+g=(p-1)%\land i+1 =\mc Z.
$%\end{eqnarray}
\paragraph{Addition of tuples:}
We can consider an $a$-tuple of numbers as a number of length $a$ in
base $p$, so addition extends naturally on it. Let us write
$\plu_{a,r}$ for the addition of $a$-tuples of $a$-ary relations of
order $r$. The idea is the same as the addition of string of bits,
with the difference that propagating overflows can be done in
different ways. The creation of an overflow at position $j$ happens
only if $C^{a,r}(\mc X_{j},\mc Y_{j})$ overflows, and then it
propagates at position $k$ if $X_{k}+Y_{k}\ge p-1$. But since, if
$X_{k}+Y_{k}> p-1$ then we have $C^{a,r}(\mc X_{k},\mc Y_{k})$, we can
consider that the overflow was created at position $k$. Hence we
consider that the only way for a overflowing bit to propagate itself
is when $\mc X_k+\mc Y_k= p-1$:
\begin{eqnarray}
  \plu_{a,r}(\ol{\mc X},\ol{\mc Y},\ol{\mc Z})\ed \bigwedge_{1\le i\le a}\ifte(\bigvee_{1\le j<i}C(\mc X_j,\mc Y_j)\bigwedge_{j<k<i}\plu_m^{a,r}(\mc X_k,\mc Y_k,p-1)),\nonumber\\
  \thent \exists \mc T^{a,r}.( \plu_m^{a,r}(\mc X_i,\mc Y_i,\mc T)\wedge \plu_m^{a,r}(\mc T,1,\mc Z_i))\nonumber\\ 
  \elset (\plu_m^{a,r}(\mc X_i,\mc Y_i,\mc Z_i))
\end{eqnarray}

% Once again, overflowing can be defined:
% \begin{eqnarray*}
% C_{a,r}(\ol{\mc X}^{a,r},\ol{\mc Y}^{a,r})\ed\neg\exists\mc \ol{Z}^{a,r}\plu_{a,r}(\mc{\ol X,\ol Y,\ol Z})
% \end{eqnarray*}

% and an addition modulo $p^a$ can be defined in the same way:
% \begin{eqnarray}
% \mathcal\plu^m_{a,r}(\ol {\mc X}^{a,r},\ol {\mc Y}^{a,r},\ol{\mc Z}^{a,r})\ed\ol{\mc X}+\ol{\mc Y}=\ol{\mc Z}\lor\nonumber\\
% \exists \ol{d},\ol e,\ol f,\ol g,\ol h,
% \ol d+\ol {\mc X}=\ol{(ap-1)}\land \ol d+\ol 1=\ol e\land \ol f+\ol {\mc Y}=\ol {(ap-1)}\land\nonumber\\
% \ol f+\ol 1=\ol g\land \ol e+\ol g=\ol h\land \ol i+\ol h=\ol {(ap-1)}\land \ol i+\ol 1 =\ol {\mc Z}
% \end{eqnarray}
\section{Relations between High-Order queries and complexity classes}\label{equ}
As stated in Section \ref{def}, we have decided to accept high-order
vocabularies. For the logic of order $r$ we accept formulae with
quantifiers of order up to $r$, but the vocabularies can contain
relations of any order. We may usually assume that the order of the
vocabulary is at most $r+1$, which is coherent with \FO{} which
contains second order relations as its input. This is because, a
relation of order $r+2$, can only be used with relation of order $r+1$
which could not be quantified, hence those relations are in the
structure, and those relations could be replaced by their truth
value without loss of generality.

\subsection{High Order and Bounded Alternating Time}

% \begin{theorem}Let $r\ge 2,j>0$ and $f\le r$. $\hoc
%   rjf=\ATIME{\exp_{2}^{r-2}(n^{O(1)})}{j}$ on structures with
%   reasonable inputs.  The same property is true over structures of
%   order $r+1$ for Turing Machine withs random access to the input.
% \end{theorem}
\begin{theorem}\label{atimeorder} For $j>0$, 
  $\hoc{r}{j}{c}=\ATIME{\exp_2^{r-2}(n^{O(1)})}{j}$ for $c\le r+1$ with
  a reasonable input, as defined in Section \ref{acc}.
\end{theorem}

This theorem is true for $c=2$ since
\cite{lauri,kolo} proved that
$\hoc{r}{j}{2}=\NTIME(\exp_{2}^{i-2})^{\Sigma_{j-1}^{\P}}$. They did
not write the ``2'' since in their definitions every formulae are over
structures of order 2.

We will then prove the theorem directly for queries over high-order
structures.
\begin{proof}
  Since $\hoc rj2\subset\hoc rjf$ , then \hoc rj2 is at least as
  expressive as the definition of \cite{lauri}, so we have this side for
  free: $\hoc{r}{j}{f}\supseteq\ATIME{\exp_2^{r-2}(n^{0(1)})}{j}$.

  We now want to prove $\subseteq$; let $\phi$ be a query in
  $\hob{r}{j}$, so then $\phi=\exists \olmc X^{r}_{1}.\forall\olmc
  X^{r}_{2}.\dots Q\olmc X^{r}_{j}.\psi$ where $\psi\in\hoa{r-1}$. We
  can begin by existentially guessing $\olmc X_{1}^{r}$, which asks us
  to write $O(\log(N(r)))=\exp_2^{r-2}(n^{O(1)}))$ bits for each
  variable of $\olmc X_{i}$. Then we can universally choose a value
  for $\olmc X^{r}_{2}$, and so on. This takes time and space
  $O(\exp_2^{r-2}(n^{O(1)}))$ and $j-1$ alternations. 

    Now everything we will do will use
  deterministic time. There are a finite number of variables, let us
  say $v$ variables, of order up to $r-1$. Hence each variables can
  take at most $N(r-1)$ values, and there are then
  $N(r-1)^{k}=\exp_{2}^{r-2}(n^{O(1)})^{k}=2^{\exp_{2}^{r-3}(n^{O(1)})\times
    k}=\exp_{2}^{r-2}(n^{O(1)})^{k}$ possible values for the $k$
  variables. Writing one of the possible values of those $v$ variables
  on the tape will take $kB(r-1)=\exp_{2}^{r-3}(n^{O(1)})$, so writing
  all of the possible tuples will take
  $\exp_{2}^{r-3}(n^{O(1)})^{k}.\exp_{2}^{r-2}(n^{O(1)})^{k}=\exp_{2}^{r-2}(n^{O(1)})^{k}$
  deterministic time and space.

  Finally we want to check the quantifier-free part of the formula,
  and it is clear that every relation, either quantified relations or
  relations of the structure of order up to $r$, can be checked in
  time at most $\exp_{2}^{r-2}(n^{O(1)})^{k}$ thanks to the
  ``acceptable encoding'' assumption. We will check those formulae at
  most $\exp_{2}^{r-2}(n^{O(1)})^{k}$ times, so we will spend at most
  $\exp_{2}^{r-2}(n^{O(1)})^{k}.\exp_{2}^{r-2}(n^{O(1)})^{k}=\exp_{2}^{r-2}(n^{O(1)})^{k}$
  times checking the quantifier-free part. If we use relations of
  order $r+1$, to check $\mc R^{a,r+1}(\olmc S^{a,r})$ we need to use
  random-access, to check if the $\olmc S$ bit of $\mc R$ is 1 or not.

  When we consider the total time, we see that it is indeed in
  $\exp_{2}^{r-2}(n^{O(1)})^{k}$, and we used $j-1$ alternations, so
  the theorem is true.
\end{proof}
Taking the union of every classes considered in Theorem
\ref{atimeorder}, we have the following corollary:
\begin{corollary}
  Over any structure, we have $\ELEMENTARY=\HO$.
\end{corollary}

\subsection{Operators on \HO}
In this section, we will prove that the properties we obtain while
adding operators to first and second order logic, relating those
logics to space complexity and deterministic time complexity, extend
naturally over \HO.

In the paper \cite{nfp}, where the nondeterministic and alternating
fixed points are introduced, a characterization of the expressivity of
first order logic with operators was given in term of ``relational
machines''. The reason is the Turing machine model implies an order
over the input, which is avoided by the relational machines, so that
they are better simulations of general first order formulae. Since in
second order we can quantify an order over the universe, and this
order then extends over high order relations there is no loss of
generality in working with Turing machines.

As explained in Subsubsection \ref{opnorm}, there is a normal form for
the formulae with operators. Every formula can be assumed to be
either like $(\TC_{\olmc X,Y}\phi)(\ol{0,\max})$ or like $(F_{P,\olmc
  X}\phi)(\ol 0)$, where $F$ is a fixed point operator and $\phi$ a
formula in $\HO$. So in this subsection we are always going to assume
that the formulae are in this form.

The table \ref{opnumber} summarizes the maximum number of steps an operator
can make without looping, and the number of bits of information
accessible at each state. There is no information about non
deterministic and alternating computation since it does not change
those numbers.
\begin{table*}\caption{Numbers of the fixed point.}
\label{opnumber}
  \centering
    \begin{tabular}{|l|l|l|}
      \hline
      $\hoa r(\P)$&Maximal number of step $\P$&Number of
      bits \\
      \hline
      $\P=\TC$&$T(r)=\ex{r-1}$&$B(r)=\ex {r-2}$\\
      $\IFP$&$C(r+1)=\ex{r-1}$&$B(r+1)=\ex {r-1}$\\
      $\PFP$&$N(r+1)=\ex r$&$B(r+1)=\ex {r-1}$\\
      \hline
    \end{tabular}
\end{table*}

\sss{Inflationary fixed point and alternating partial fixed point}
It is already known that $\P=\FO(\IFP)$ over ordered structures, and similarly
\EXP=\SO(\IFP).  In \cite{nfp} it was proved that $\FO(\NIFP)${} is \NP{} over first
order with an order relation.  They are special cases of the theorem:

\begin{theorem}\label{tifp}
  Over reasonable input we have $\A\SPACE(\exp^{r-1}_2(n^{O(1)}))=\hod
  rj(\APFP)=\hoa{r+1}(\IFP)=\DTIME(\exp^{r}_2(n^{O(1)}))$.
\end{theorem}

The article \cite{DETH} proved
$\hoa{r+1}(\IFP)=\DTIME(\exp^{r}_2(n^{O(1)}))$ assuming an order over
the structure and a vocabulary of order 2. Our proof is similar, but
we begin by constructing order and arithmetic relations thanks to
second-order relation.

\begin{proof}
  It has been proven in \cite{alternation} that when $f$ is a function
  greater than the logarithm, $\A\SPACE(f)=\DTIME(2^{O(f)})$ hence
  $\A\SPACE(\exp^{r-1}_2(n^{O(1)}))=\DTIME(\exp^{r}_2(n^{O(1)}))$.

 \paragraph{ Proof of
   $\hoa{r+1}(\IFP)\subseteq\DTIME(\exp^{r}_2(n^{O(1)}))$}
 Let $\phi\in\hoa{r+1}(\IFP)$, such that $\phi=(IFP_{\ol x,P}\psi)(\ol
 y)$. Suppose that $\ol x=x_1,\dots,x_n$, then there are
 $T(r+1)=\exp^{r+1-1}_2(n^{O(1)})$ sets of tuples of relations
 equivalent to $\ol x$, hence we find the fixed point after at most
 $\exp^{r+1-1}_2(n^{O(1)})$ steps. Since $\psi\in HO^r+1$ we know that
 $\psi\in\Sigma_{j}\TIME(\ex{r+1-2})^{P}$ for some $j$ where ``$^P$''
 is an oracle in $P$. This class is a subset of
 $\DTIME(\exp_2^{r+1-1}(n^{O(1)}))^P$, and since there are at most
 $\exp^{r+1-1}_2(n^{O(1)})$ elements in $P$ it can still be coded with
 a string of bits, and then checked in time
 $\exp^{r+1-1}_2(n^{O(1)})$. Since in time $\exp^{r+1-1}_2(n^{O(1)})$
 there are at most $\exp^{r+1-1}_2(n^{O(1)})$ queries to the oracle,
 then checking $\psi$ takes time
 $\exp^{r+1-1}_2(n^{O(1)})^2=\exp^{r+1-1}_2(n^{O(1)})$.

  During the $i$th step we will check for every tuple of relation
  $\ol z$ if $z\in P_i$, applying $\psi$ with input
  $P_{i-1}$. Since there are up to $N(r+1)=\exp^{r+1-1}_2(n^{O(1)})$
  possible relations, each step will take time
  $\exp^{r+1-1}_2(n^{O(1)})\times\exp^{r+1-1}_2(n^{O(1)})=\exp^{r+1-1}_2(n^{O(1)})$. Finally,
  since there are at most $\exp^{r+1-1}_2(n^{O(1)})$ steps, the entire
  computation will take time
  $\exp^{r+1-1}_2(n^{O(1)})\times\exp^{r+1-1}_2(n^{O(1)})=\exp^{r+1-1}_2(n^{O(1)})$,
  which ends this side of the proof.

 \paragraph{ Proof of
   $\hod   rj(\APFP)\subseteq\hod{r+1}j(\IFP)$}
  Let $\xi\in\hoa r(\APFP)$, $\xi=(\APFP_{P^{r+1},\olmc
    X^{r}}\phi,\psi)(\olmc Y)$

  We are going to use an inflationary fixed point to create the tree
  $T=T_{\phi,\psi}$. We will associate the label of every node to its
  path in $T$.

  Since there is at most $B(r+1)=\ex{r+1}$ values that $P$ can take
  then there is at most $2^{\ex{r-1}}=\ex{r}$ paths of such length.
  But it is correct since $Q$ can also take $C(r+2)=\ex{r}$ values.

  Since by Claim \ref{AFP} the tree $T$ can be cut once we met twice
  the same relation in a branch, and that there is at most $\ex{r+1}$
  relations, we can cut the tree at depth $\ex{r+1}$, hence using a
  simple fixed point is not a problem. 

%   To be able to do the fixed point, we will use one first order
%   variable, such that $Q(0,p$ has a relation with this variable to $0$ a
%   path and the empty relation if and only if the value of this path
%   was already found. We use the same trick to distinguish the nodes
%   that are fixed points.

  Then with a second fixed point, we recursively calculate the output
  of the circuit. We consider the leaves that are not a fixed point to
  be the relation $\bot$, the leaves which are fixed points we
  consider the relation in their label. Then we do union and
  intersection of the gates when we know their children's value.

% \begin{eqnarray}
%   \xi'=(\IFP_{P,\olmc X}\theta)(\olmc Y)\\
%   \theta\ed\\
%   \phi'( \olmc Y)\ed(\IFP_{P^{r+2},b,\olmc C^{r+1},\olmc X^{r}}\psi')(\ol\max,\olmc Y)\\
% \psi'\ed (b=1\land \olmc C=\olmc X=0)\lor\nonumber\\
% (P(1,\ol{\mc{C}/2},\ol0)\land\neg P(1,\olmc C,\ol0)\land((b=1\land X=0) \lor\nonumber\\
% (\ifte \text{even}(C)\nonumber\\
% \thent (\phi[P/P(1,\ol{\mc C/2},.)](\olmc X)\land \psi[P/P(1,\ol{\mc C/2},.)](\olmc X))\nonumber\\
% \elset (\phi[P/P(1,\ol{\mc C/2},.)](\olmc X)\lor \psi[P/P(1,\ol{\mc C/2},.)](\olmc X))))
% \end{eqnarray}
% We have $\olmc C$ equivalent to $\olmc X$ of one order higher, plus
% one bit. We use the higher order to be sure to have time to
% enumerate every relation before we find some fixed point, and we need
% one more bit because we could imagine to see twice the same
% configuration, first as an union, then as an intersection. Thanks to
% claim \ref{AFP} we see that it is not a problem that we wait for $C$
% to be the maximum to give the answer.
% \paragraph{$\supseteq$} Let $\xi\in\hoa {r+1}(\IFP)$,
% $\xi=(\IFP_{P^{r+2},\olmc X^{r+1}}\psi)(\olmc Y)$
% \begin{eqnarray}
%   \xi'\ed(\APFP_{P^{r+1},\olmc X^{r}}\phi,\psi)(\olmc Y)
% \end{eqnarray}

 \paragraph{ Proof of
   $\A\SPACE(\exp^{r-1}_2(n^{O(1)}))\subseteq\hod  rj(\APFP)$}
 The proof for $r=1$ was given in \cite{nfp}. The same proof works for
 $r>1$, except that we can construct an arbitrary order as explained
 above. 

% \paragraph{Proof of $\DTIME(\exp^{r}_2(n^{O(1)}))\supseteq\hoa{r+1}(\IFP)$}
% We will use the fact that $\A\SPACE(f(n))=\DTIME(O(1)^{f(n)})$
% when $f>\log$, hence that
% $\A\SPACE(\exp^{r-2}_2(n^{O(1)}))=\DTIME(\exp^{r-1}_2(n^{O(1)}))$. We
% will now prove $\hoa{r}(\IFP)=\A\SPACE(\exp^{r-2}_2(n^{O(1)}))$.

% As we already know, we can encode a configuration of a TM in
% $\SPACE(\exp_{2}^{r-2}(n^{O(1)}))$ using relations of order $r$. We also
% assume that we can tell whether one configuration is the successor of
% another one. If a configuration is universal (existential), it is
% accessible if all (at least one) predecessors are accessible, so the
% set of accessible configuration can clearly be defined as a
% inflationary fixed point.

% We can assume without loss of generality that the initial state is
% universal and has no predecessors. This way it is accessible, and
% there is only one accepting configuration, with empty tape and the
% head at the beginning of the tape. Then the formula will check if the
% description of this configuration is accessible.
\end{proof}

Once again, accepting that the input contains high order relations does
not change the expressivity, if we consider acceptable input, and that the
input size is the size of the structure and not the size of the
description. And since we have time $\exp_{2}^{r-1}(n^{O(1)})$ and not
$\exp_{2}^{r-2}(n^{O(1)})$, we can even check element of relation of
order $r+1$.

\sss{(Non)deterministic partial fixed point, Transitive
  closure, Alternating inflationary fixed point and Space complexity}
It is already known that
$\FO(\AIFP)=\FO(\NPFP)=\FO(\PFP)=\SO(\TC)=\PSPACE$ over ordered
structures. These equality are special cases of  Theorem
\ref{space}:
\begin{theorem}\label{space}
  Over reasonable input we have $\hoa
  r(\AIFP)=\hoa{r}(\NPFP)=\hoa{r}(\PFP)=\hoa
  {r+1}(\TC)=\SPACE(\exp_{2}^{r-1}(n^{O(1)}))$ .
\end{theorem}
We are going  to transform formulae from one
formalism to another one without going through machines, giving a
pattern of algorithms for the transformation.
There will be an exception for $\hoa r(\NPFP)$ that we only know how
to transform into a space bounded TM, the equality using Savitch's
theorem \cite{Savitch}.

The result for $\AIFP$ is not a surprise if we consider that $\IFP$ is
time and $\A$ is alternations, so that this theorem is similar to
$\A\TIME(f)=\SPACE(f)$.

\begin{proof}Proof of $\hoa r(\AIFP)\subseteq\hoa {r}(\PFP)$: Let $\xi\in\hoa
  r(\AIFP)$, $\xi=(\AIFP_{P^{r},\olmc X^{r-1}}\phi,\psi)(\olmc Y)$.
  Then the fixed point can be obtained with at most $C(r)=\ex {r-2}$
  iterations since it is inflationary, and there is at most
  $2^{C(r)}=\ex{r-1}$ paths.

  We are going to transform $\xi$ in an $\hoa r(\PFP)$ formula.  In
  $\PFP$ we can do $T(r)=\ex {r-1}$ steps, which is enough to test
  every path. We will make a relation $Q$ which has 3 arguments. The
  first one is a path $p$ in the tree $T_{\phi,\psi}$, {\it i.e.} a
  string of bits such that the $i$th bit is 0 if the $i$th step in
  $\AIFP$ was $\phi$ else 1. When the second argument is 0 then the
  third argument is the relation $P_{p}$, else if the second argument
  is 1 then the third argument is 0 to mean that the relation $P_{p}$
  was defined.

  For first step, we let $Q(0,0,0)$ and $Q(0,1,0)$ be true. During the
  next step if $Q(C/2,1,0)$ is true then we set $Q(C,1,0)$ and
  $Q(C,0,\phi(P_{C/2}))$ to be  true. Finally we end the
  computation when for every $C$, $Q(C,1,0)$ is true, then we have
  the $\ex {r-2}$ level of the tree $T_{\phi,\psi}$, and every
  relation $P_{p}$ can be checked in $Q(p,0,.)$. We check if there is
  one of those relations that is a fixed point, and that contains $\olmc
  Y$. If yes, we accept $Q(2,0,0)$, else $Q(2,0,1)$. We can not miss a
  fixed point; since it is inflationary, we see it at or before step $\ex
  {r-2}$, and if we discovered the fixed point sooner, it is not a
  problem if we continue to apply $\phi$ or $\psi$. (We still have got
  the fixed point, by the very definition of fixed points.)

  If there is $Q(2,0,b)$ with $b\in\{0,1\}$ which is true, then we
  accept only $Q(2,0,b)$ so we indeed have got a fixed point, and we
  accept only if $(2,0,0)$ is in the output of this $\PFP$. This ends
  the proof.

\paragraph{Proof of  $\hoa {r}(\PFP)\subseteq\hoa {r+1}(\TC)$:}
Let $\phi\in\hod {r}j(\PFP)$. By the normal form property we can
assume that
$\phi=(\PFP_{P^{r+1},\olmc{X}^{r}}\psi(P,\olmc{X}^{r}))(\ol 0^{r})$. We
also assume that $\olmc X$ contains only $r$th order variable and
$\psi$ is in \hod {r+1}{r+1}. Then
\begin{eqnarray}\phi'\ed(\TC_{x,\mc{P}^{r+1},y,\mc{Q}^{r+1}}\psi')(0,\ol 0,1,\ol{\max})\\
  \psi'\ed x=0\land\mc (\ifte P=\psi(\mc P,.)\thent(\ifte \mc P(0)\thent (y=1\land \mc Q=\ol{\max})\nonumber\\
  \elset \bot)\elset (y=0\land \mc Q=\psi(\mc P,.)))
\end{eqnarray}
Here  $\mc P=\psi(\mc P,.)$ is syntactic sugar for $\forall \olmc Z^{r}
(\mc{P(\ol Z)\Leftrightarrow \phi(P,\ol Z)})$.

\paragraph{Proof of $\hoa {r+1}(\TC)\subseteq\hoa r(\AIFP)$}
Let $\xi\in\hoa{r+1}(\TC)$. We suppose that $\xi$ is in normal form,
hence $\xi=(\TC_{\olmc {XY}^{r+1}}\psi)(0,\max)$ with
$\psi\in\hoa{r+1}$. Let us say that $P_{0}=0$ and $P'_{0}=\max$, so that the
transitive path from $P_{0}$ to $P'_{0}$ can take up to $T(r+1)=\ex r$
steps. We are of course going to do a divide and conquer method,
existentially guessing the middle $P''_{0}$ of the path, and
universally checking both sides, that there is both a path from
$P_{1}=P_{0}$ to $P'_{1}=P''_{0}$, and from $P_{1}=P''_{0}$ to
$P'_{1}=P'_{0}$, and so on.  Hence we need to make at most
$\log(T(r+1))=\ex {r-1}$ guesses. For each choice there are
$T(r+1)=\ex r$ possibles choices. In \AIFP{} we can only choose one
element of two ($\phi$ or $\psi$) so we will need to guess the
relation in the middle of the path bit by bit, so it will take
$\log(T(r+1))=\ex {r-1}$ guesses of bit; we use a counter to find when
we have guessed every bit, while there are bits to guess the universal
choice does not do anything. In total this makes
$\log^{2}(T(r+1))=\ex {r-1}$ existential guesses and $\log(T(r+1))=\ex
{r-1}$ universal ones. This is possible in $\hoa r(\AIFP)$.

Finally we existentially guess when the path is one step long, then we
just check that indeed $\psi(P,P')$ is true.
\paragraph{Proof of $\hoa{r}(\PFP)\subseteq\hoa{r}(\NPFP)$:} This is trivial,
it suffices to transform a formula of $\hoa r(\PFP)$ into normal form,
so that no negation are applied to the operator, and then transform
$\PFP_{\phi}$ to $\NPFP_{\phi,\phi}$.

\paragraph{Proof of $\hoa{r}(\NPFP)\subseteq\SPACE(\exp_{2}^{r-1}(n^{O(1)}))$ :}
Let $\phi\in\hoa{r}(\NPFP)$, such that $\phi=(NPFP_{\ol
  x^{t},P^{t}}\psi,\xi)(\ol y)$. We are going to give an algorithm in
$\NSPACE(\exp_{2}^{r-1}(n^{O(1)}))$, which can be simulated in
$\SPACE(\exp_{2}^{r-1}(n^{O(1)}))$ by Savitch Theorem.

Suppose that $\ol x^{t}=x_1,\dots,x_n$. Then there are
$T(r)=\exp^{r-1}_2(n^{O(1)})$ sets of tuples of type $t$, and hence
writing a value of $P_{i}$ takes $\exp^{r-1}_2(n^{O(1)})$ bits. We
begin by writing $P_{0}=\bot$, and we loop so that when we know
$P_{i}$ we guess if it is a fixed point, then we look if
$\psi(P_{i})=\xi(P_{i})=P_{i}$ and if $P_{i}(\ol j)$; if yes we
accept, else we reject. Else we guess if we need to apply $\phi$ or
$\psi$ to obtain $P_{i+1}$, where the $j$th bit is $1$ if the $j$th
relation equivalent to $\ol x$ is true. We can then loop over every
possible relation of type $t$ to see if it is in $P_{i+1}$,
enumerating these relations take space $T(r)=\exp^{r-2}_2(n^{O(1)})$,
and it is already known that testing $\psi\in \hoa{r}$ is in
$\A\TIME(\exp^{r-2}_{2}(n^{O(1)}))\subseteq
\PSPACE(\exp^{r-1}_{2}(n^{O(1)}))$. Once $P_{i+1}$ is known we can
forget $P_{i}$, so there is no need of more space.

\paragraph{Proof of $\SPACE(\exp_{2}^{r-1}(n^{O(1)}))\subseteq\hoa{r}(\PFP)$:}
As we already know, we can encode a configuration of a TM in
$\DTIME(\exp_{2}^{r-1}(n^{O(1)}))$ and using space
$(\exp^{r-1}_2(n^{O(1)}))$ using relations of order $r+1$; of course
our relation will be $P$. We now only need to be able to decide if one
configuration is the successor of another one.

If we encode a configuration as a string of bits, 00 for 0, 01 for 1,
and 1$x$ for the head of the Turing machine in state $x$, then to
decide the value of the bit $i$ at time $t+1$, we only need to look at
up to $\log |x|+5$ bits on the left and on the right of a bit at time
$t$. Since we have a ``$\suc$'' relations over high-order relation we can
easily do it in $\hoa{r}$ (because $\log|x|$ is a constant for a given
TM). This let us speak of the next step of the Turing machine.

% We already show in section \ref{spaceinho} that we can transform the
% input over second order into an input over a relation of order $r$,
% we will just use the same code.

We can assume without loss of generality that there is only one
accepting configuration, with empty tape, and that the Turing machine
loops on this configuration. Then the formula will check if the
description of this configuration is accessible.
\end{proof}
% As in section \ref{space}, we see that Turing Machine can be
% simulated by a formula over a second-order structure, but on other
% hand, every formula over an high-order structure can still be
% simulated by the same class of space complexity TM.

Once again, accepting that the input contains high order relations does
not change the expressivity, if we consider only acceptable input, and that the
input size is the size of the structure and not the size of the
description. And since we have space $\exp_{2}^{r-1}(n^{O(1)})$ and
not $\exp_{2}^{r-2}(n^{O(1)})$, we can even check elements of relations
of order $r+1$.

\sss{Nondeterministic inflationary fixed point}
In \cite{nfp} it was proved that $\FO(\NIFP)${} is \NP{} over first
order with an order relation. This is a special case of Theorem
\ref{tnifp}

\begin{theorem}\label{tnifp}
  Over reasonable input we have $\hod rj(\NIFP)=\hoc{r+1}1j$.
  % $\hoa r(\NIFP)=\NTIME(\ex{r-1})$ over reasonable input.
\end{theorem}
  \begin{proof}$\subseteq:$ Let $\xi=(\NIFP_{P^{r+1},\olmc
      X^{t,r}}\phi,\psi)(\olmc Y)$ with $\phi,\psi\in \hoa r$. Then we
    will existentially guess a relation $Q$ whose type is a pair of
    $\olmc X$'s type. The first half of the arguments is a time-stamp,
    such that $Q(\olmc C,.)$ is the relation $P_{C}$ where $\olmc C$
    is considered as a number.

  Since $\xi$ is an inflationary point, it can take at most
  $C(r+1,a)=\ex {r-1}$ iterations; since the counter, which consists of
  variables of order $r$, can count up to $T(r)=\ex {r-1}$ we can
  indeed encode every steps in one relation.

  We then just need to check if $\olmc Y$ is in $Q(\olmc C,.)$ for
  some $\olmc C$ such that $Q(\olmc C,.)$ is a fixed point for both
  $\phi$ and $\psi$.
\begin{eqnarray}
  \xi'=\exists Q^{tt,r}.\{(\neg \exists \olmc X^{t}.Q(\ol
  0,\olmc X)) \land
  (\exists \olmc C^{t}. Q(\olmc C,\olmc Y)\land
  \forall \olmc X^{t} \phi(Q(\olmc C, \olmc X))\Leftrightarrow Q(\olmc C, \olmc X)\Leftrightarrow
  \psi(Q(\olmc C, \olmc X))) \land\nonumber\\\forall \olmc T^{t}\olmc X^{t}.Q(\olmc
  T,\olmc X)\Leftrightarrow
  (Q(\ol {\mc T-1},\olmc X)\lor\phi[P/Q(\ol{\mc T-1},.)](\olmc X)\lor\psi[P/Q(\ol{\mc T-1},.)](\olmc X))\}
\end{eqnarray}
% $\subseteq$, the proof is almost the same than the one of
% \ref{tifp}, except that now on any fixed points relation, we should
% guess each time which formula we apply. Since it is inflationary,
% the same time bound still applies, and since there is no negation
% applied to the operator, by definition \ref{dnpfp}, existential
% guess are enough.
\paragraph{Proof of $\supseteq$:}
Let $\phi=\exists \olmc Q^{r+1}. \psi$ with $\psi\in\hoa r$. We will
nondeterministically guess every bit of $Q$; there are $C(r+1)=\ex
{r-1}$ such bits and we can do $T(r)=\ex {r+1}$ steps in an
inflationary fixed point. 

We will create a relation $P$ that takes three arguments. The second
one is a time-stamp $\olmc C$. If the first argument is 0 then the
last argument is the string of bits that we are constructing. Else if
the first argument is $1$ then the third argument is $0$; this means
that the string of bits at time $\olmc C$ was already defined.

When $\olmc C=0$ we must have $\olmc X=0$, and when $\olmc C>0$, if
$\ol{\mc C-1}$ is defined and $\olmc C$ is not, then the values of
$\olmc X$ is either multiplied by 2, in $\psi'$, or by $2$ and
incremented by $\xi'$. Finally, when the string of bits is completed,
we check if $\psi$ is true when $Q(\olmc X)$ is replaced by
$P(0,\max,\olmc X)$.

If it is true, we accept the arguments $(2,0,0)$, else nothing. Since
nothing else changes, this is a fixed point, and $\phi'$ will be true
if and only if $\psi$ is verified by this string of bits.
\begin{eqnarray}
  \phi'\ed(\NIFP_{P,b,\olmc {C,X}}\psi',\xi')(2,\ol 0,\ol 0)\\
  \psi'\ed b=1\land \olmc C=\olmc X=0\lor \ifte \exists \olmc X'. P(0,\max,X')\nonumber\\
  \thent (\ifte \psi[P/P(0,\max,.)] \thent (b=2\land \olmc X=\ol0\land C=\ol0)\elset\bot)\nonumber\\
  \elset(P(1,\ol{\mc C-1},0)\land\neg P(1,\ol{\mc C},0)\land\nonumber\\
  ((b=0\land P(0,C,.)=2P(0,C-1,.))\lor(b=1\land \olmc X=0)))\\%\olmc X=2\olmc X'
  \xi'\ed P(1,\ol{\mc C-1},0)\land\neg P(1,\ol{\mc C},0)\land\nonumber\\
  ((b=0\land P(0,C,.)=2P(0,C-1,.)+1)\lor(b=1\land \olmc X=0))%\olmc X=2\olmc X'+1
\end{eqnarray}
\end{proof}

We think that this is an equality (at least for $r=1$ it is one), but the
other side of the relation seems harder to prove.
\subsection{Horn and Krom formulae }\label{horn}
Another important result in descriptive complexity theory is that
$\P=\SO(\HORN)$ and $\NL=\SO(\KROM)$. We will discuss the problem of
extending these results to higher-order.

\begin{definition}[Horn and Krom formula]A literal is an atomic predicate
  or its negation, the first one is called a positive literal and the last
  one a negative literal. A disjunction of literals is a clause, and a
  conjunction of clauses is a quantifier free formula in conjunctive
  normal form (CNF). A CNF formula is then a formula $\phi=Q_1\mc
  X_1^{2}\dots Q_n\mc X^{2}_n\forall \ol x\psi$, where the $Q$ are
  quantifiers and $\psi$ is a quantifier-free CNF formula.

  A Horn formula is a CNF formula such that in each clause there is
  exactly one positive quantified literal. A Krom formula is a CNF
  formula such that in each clause there are at most two literals.
\end{definition}

Over second order, the proof of the equality begins by proving that
those classes have a normal form where every second order quantifier
in existential. Over higher order, it is not easy to see what this
normal form would be. For example in \hoa{3} we can not require
the second order quantifiers to both all be universal and all be existential. And if we
accept the first order to be also existential then problems like
``clique'', which are known to be NP-complete, can be coded in
\SO(\HORN), so finding the good restriction over quantifiers is
mandatory to have an interesting result.

\subsection{Monadic High-Order Logic (MHO)}
\label{sec:mho}
Monadic Second Order $\MSO$ is a well-studied logic, we intend to
study the monadic restriction of logic of order at least 3, as we will
see the theory is really different.

\begin{definition}
%   A \emphdex{monadic} relation of order $r$ is a unary (1-ary) relation
%   of order $r$. An equivalent definition is that
  The set of monadic relations of order $r\ge 1$ is the $(r-1)$th
  power set of the universe, $\mc P^{r-1}(A)$; where we define $\mc
  P^{0}(E)=E$ and $\mc P^{r}(A)=\mc P(\mc P^{r-1}(E))$ and $\mc P$ is
  the usual power set operation. 

  The \emphdex{Monadic High-Order Logic of order $r$ } (\mhoa{r}) is
  defined as the subset of queries of $\hoa{r}$ where all quantified
  relations are monadic. The definitions of $\mhod{r}{f},\mhob{r}{j}$
  and $\mhoc{r}{j}{f}$ are straightforward extensions of the \HO{} and
  $\Sigma$ definitions.
  % Since the input structure can contains many-ary relations, and we
  % want to accept many-ary relation in
\end{definition}
The definition only restricts the arity of quantified relation, and so
the vocabulary of a formula may contain many-ary relations.

It is well known that one of the main problem with $\MSO$ is that one
can not create an order over the structure. But in Monadic Third Order
one can quantify a set of the form $\{[0,i]|0\le i< n\}$ and use this
as a linear order over the structure. This let us create addition with
the set $\{\{a,b,c\}|a+b=c\}$ and multiplication with the set
$\{\{a,b,c\}|a\times b=c\}$, hence we can define a ``bit'' predicate
and simulate Turing Machine. 

It is important to realize that Theorem \ref{atimeorder} assumed that
we can increase the arity to obtains more space and time. Since we can
not do it anymore we see that the big $O$ is not anymore in the top of
the tower of exponential, but in the second floor. Hence we
obtain similarly  Theorem  \ref{mho_t}

\begin{theorem}\label{mho_t}
  $\mhoc rjc=\ATIME{2^{O(\exp^{r-3}_{2}(n))}}j$ for $c\le r+1$ with a
  reasonable input.
\end{theorem}

% It only remains to calculate the space and time needed by a Turing
% Machine to simulate an $\mhob{r}{j}$ formula, and what . Theorem \ref{atimeorder}
% use the assumption that we can increase the arity of the variables if

% \begin{notation}
%   Since we only quantify unary relations, we will write $\mc X^{p}$
%   instead of $\mc X^{((\dots(\iota)\dots))}$ with $p-1$ pairs of
%   parentheses.
% \end{notation}

% Using Lemma \ref{size} we see that we need
% $B(r,1)=\texp^{r-2}_{2}(n,1)=\exp^{r-2}_{2}(n)$ bits to describe one
% monadic variable of order $r$, and if we quantify $c$ variables of
% order $r$ we may have to try the quantifier-free part of the formula
% up to
% $(N(r,1))^{c}=(2^{\texp^{r-2}_{2}(n,1)})^{c}=2^{c\exp^{r-2}_{2}(n)}=2^{O(\exp^{r-2}_{2}(n))}$. Hence
% a Turing Machine simulating an $\mhoa{r}$ formula takes at least
% $O(\exp^{r-2}_{2}(n))$ space and $2^{O(\exp^{r-2}_{2}(n))}$ time. It
% is important to notice that the position of the big $O$ is different,
% it is not up in the tower of the exponentiels anymore but in first or
% second floor.

\section{Conditional relations among the classes}\label{cond}
In this section, we will discuss theorems of the form ``If $A=B$
then $C=D$'' where $A, B, C $ and $D$ are complexity classes or
theories over finite models.  Most results use a padding argument
or are corollaries of theorems known on lower complexity classes. What
will be more interesting is to study the results that seems intuitive
but that we do not know how to prove.

There are conjectures in high complexity classes which seem to be copy
of theorem over polynomial classes, we will explain why the known
proof for polynomial classes fails on higher classes.

We are going to work mostly with Turing Machine, and we will also
translate the results are descriptive complexity's theorem or
question.

We also should emphasize the fact that when we do not explicitly state
any assumptions over the function classes, then they could contains
only one function, hence we also obtain theorem over complexity time
bounded by a function.

\subsection{The $r$th exponential hierarchy}
It is known that $\SO=\PH$, the polynomial time hierarchy, and
$\SO_{j}=\hob{2}{j}=\Sigma^{\P}_{j}$ is the $j$th level of the
polynomial hierarchy. We are going to extend this hierarchy to higher
order.
\begin{definition}[$r$th exponential hierarchy]
  Let $\hoa{r+2}$ be the $r$th exponential hierarchy, and
  $\hob{r+2}{j}$ be the $j$th level of the $r$th exponential
  hierarchy.
\end{definition}
We choose the name such that the (alternating) time of $r$th
exponential hierarchy has $r$ exponential under the $n$.  We have the
polynomial hierarchy as the $0$th exponential
hierarchy. Our definition is different from the ``Exponential
hierarchy'' of \cite{pap} in that his hierarchy is
$\bigcup_{i\in\mathbb N}\TIME(\exp_{2}^{i}(n^{O(1)}))$, and in each of
our levels we also consider alternations.

\begin{definition}[Collapsing]
  For a class of function $C$ we say that $C$ collapses to the $j$th
  level if $\forall k\ge j$, $\ATIME{C}{j}=\ATIME{C}{k}$.  By extension we
  say that \hoa r (resp. \hod rf) collapses to the $j$th level if for
  all $k\ge j$ $\hob{r}{j}=\hob{r}{k}$ (resp.  $k\ge j$
  \hoc{r}{j}{f}=\hoc{r}{k}{f}).
\end{definition}
% \subsection{Machines with bounded alternations}
% We are first going to give results for the collapsing of alternating
% classes, and then see how they propagate to classes of formulae.
\subsection{General classes of functions}
\begin{lemma}\label{lem:alt}
  Let $C$ be a class of function, and $j\ge 0$, if
  $\ATIME{C}{j}=\ATIME{C}{j+1}$ then
  $\ATIME{C}{j}=\coATIME{C}{j}=\ATIME{C}{j+1}=\coATIME{C}{j+1}$.
\end{lemma}
\begin{proof}
  The proof is almost identical to the one of the polynomial
  hierarchy, which is the special case $C=n^{O(1)}$. If
  $\ATIME{C}{j}=\ATIME{C}{j+1}$ then their complement are also equals,
  so we have $\coATIME{C}{j}=\coATIME{C}{j+1}$ hence
  $\coATIME{C}{j}\subseteq\ATIME{C}{j+1}=\ATIME{C}{j}\subseteq\coATIME{C}{j+1}=\coATIME{C}{j}$.
\end{proof}
\begin{theorem}\label{atimehigher}
  Let $F$ and $G$ be classes of functions such that for all $f\in F$
  there exists a function $h_{f}$ computable in time $f$(resp. space
  $f$, resp. space $\log\circ f$) and $g_{f}\in G$ such that
  $f(n)=O(g_{f}(h_{f}(n)+n))$ and for all $g'\in G$ there exists
  $f'\in F$ such that $g'(h_{f}(n)+n)=O(f')$. Let $0\le j<k$ and
  assume that $\ATIME{G}{j}=\ATIME{G}{k}$ then
  $\ATIME{F}{j}=\ATIME{F}{k})$ (resp. assume that
  $\ATIME{G}{j}=\SPACE(G,k)$ then $\ATIME{F}{j}=\SPACE(F)$, resp. assume
  that $\SPACE(\log(G))=\TIME({G})$ then
  $\SPACE(\log(F))=\TIME({F})$).
\end{theorem}
\begin{proof}
  Let $f\in F$ and $L$ a language decided by a TM
  $M\in\ATIME{f}{k}$(resp. $\SPACE(f)$, resp $\TIME({f})$) and let
  $L'=\{x1^{h_{f}(|x|)}|x\in L\}$. It can be decided by a TM $M'\in
  \ATIME{g_{f}(n)}{k}$(resp. $\SPACE(g_{f})$, resp. $\TIME({g_{f}})$)
  which tests whether there is a correct number of $1$ and then
  simulates $M$ (it is possible in our bound since $f\subseteq
  O(g_{f}(h_{f}(n)+n))$ and $h_{f}$ is constructible  in
  $\TIME(f)$), hence by our assumption there is $g'\in G$ such that
  $L'$ can be decided by a TM $M''\in\ATIME{g'(n)}{j}$ (resp. id.,
  resp. $\SPACE(\log\circ g)$). Then $L$ can be decided by a TM $M'''$
  which, on input $x$, writes down $X=x1^{h(|x|)}$, which takes time
  $O(f)$ (resp. space $O(f)$, resp time $\log\circ f$), and then
  simulates $M''$ on $X$, which takes $g'(h_{f}(n)+n)$, and by
  hypothesis there exists $f'\in F$ such that
  $f+g'(h_{f}(n)+n)=O(f')$(resp. id, resp. $\log\circ
  f+\log(g'(h_{f}(n)+n)=O(\log\circ f'))$), hence we indeed have
  $\ATIME{F}{k}\subseteq\ATIME{F}{j}$
  (resp. $\ATIME{F}{k}\subseteq\SPACE(F)$,
  resp. $\SPACE(F)\subseteq\TIME(F)$). The proof of $\supseteq$ is
  trivial since $j<k$.
\end{proof}
\begin{corollary}
  Let $f,g$ be integer functions such that there exists a function $h$,
  computable in time $O(f)$, such that $f=\Theta(g(h(n)+n))$. Then for
  all $0\le j<k$ $(\ATIME{g}{j}=\ATIME{g}{k}$ implies
  $\ATIME{f}{j}=\ATIME{f}{k})$, $\ATIME{g}{j}=\A\SPACE(g)$ implies
  $\ATIME{f}{j}=\A\SPACE(f)$ and $\TIME(g)=\SPACE(\log\circ g)$ implies
  $\ATIME{f}{j}=\A\SPACE(f)$.
\end{corollary}
It is surprising that we do not know how to prove that if
$\ATIME{C}{j}=\ATIME{C}{j+1}$ then $C$ collapses to level $j$. But we
think that it must be true, or at least that it would be really hard
to prove it to be false. First because if it was false it would imply
$\P\subsetneq\NP$, and also because it would be surprising that, for
some complexity classes, having $j$ or $j+1$ alternations is as
expressive, but having $j+2$ alternations is strictly more expressive.

% The difference between Theorems \ref{little_collapse} and
% \ref{atimehigher} is that, in the first one, the hypothesis is only
% about classes of small functions (at most polynomial) and the
% implied result is about every level higher than $j$ for any
% exponential level, whereas in the second one, the hypothesis is
% about any kind of classes of function and the implied result is only
% for the results of the hypothesis.

\begin{lemma}\label{hocollapse}
  Let $2\le r< p$ and $0<j<k$.  Then $\hoc{r}{j}{2}=\hoc{r}{k}{2}$ implies
  that $\hoc{p}{j}{2}=\hoc{p}{k}{2}$, $\hoc rj2=\hod r2(\TC)$  implies
  $\hoc pj2=\hod p2(\TC)$ and  $\hod r2(\IFP)=\hod r2(\TC)$  implies
  $\hod p2(\IFP)=\hod p2(\TC)$.
\end{lemma}
\begin{proof}
  Let $F=\exp^{r-2}_{2}(n^{O(1)})$ and $G=\exp^{p-2}_{2}(n^{O(1)})$,
  the condition of Theorem \ref{atimehigher} are respected since, for
  all $f\in F$ $g_{f}=\exp^{p-2}_{2}(n)$, $h_{f}=\exp_{2}^{p-r}(n)$,
  we have $f(n)=O(g(h(n)+n))$ and for all $g'\in G$ let
  $f'=g'(h_{f}(n)+n)$ it is easy to see that $f'\in F$ hence
  $g'(h_{f}(n)+n)=O(f'(n))$.

  By Theorem \ref{atimeorder} $\hoc{r}{j}{2}$ is equal to
  $\ATIME{\exp_{2}^{r-2}(n^{O(1)})}{j}$ and by Theorem \ref{space}
  $\hod r2(\TC)$ is equal to $\SPACE(\ex {r-2})$. Then the corollary
  is just a translation of Theorem \ref{atimehigher} in a descriptive
  complexity setting.
  % , the assumption $=\ATIME{\exp_{2}^{r-2}(n^{O(1)})}{k}$, then by
  % Theorem \ref{atimehigher},
  % with%we have  $\ATIME{\exp_{2}^{p-2}(n^{O(1)})}{j}=\ATIME{\exp_{2}^{p-2}(n^{O(1)})}{k}$ hence by Theorem \ref{atimeorder}, $\hoc{p}{j}{2}=\hoc{p}{k}{2}$.
\end{proof}

\subsection{Polynomial hierarchy and exponential hierarchies}
First we are going to prove that hypothesis on the polynomial
hierarchy and polynomial space imply results on the exponential
hierarchy. Hence we may prove some interesting result on polynomial
classes by proving them in exponential hierarchy.

\begin{theorem}\label{little_collapse}
  Let $D$ be a class of functions which contains at least every linear
  function and let $C$ be a class of time-constructible functions
  closed under addition and such that $\forall g\in D,f\in C(g\circ
  f\in C)$. If $\ATIME{D}{j}=\ATIME{D}{j+1}$ or
  $\ATIME{D}{j}=\coATIME{D}{j}$ then $\forall k\ge
  j,\ATIME{C}{k}=\coATIME{C}{k}=\ATIME{C}{j}$ and if
  $\ATIME{D}{j}=\SPACE(D)$ then $\ATIME{C}{j}=\SPACE(C)$.
\end{theorem}
Here we use a definition of TM with one reading tape and one working
tape, this way the linear time function can at least verify their
bounds.
\begin{proof}
  The first assumption implies the second one by Lemma \ref{lem:alt},
  hence we are only going to suppose that
  $\ATIME{D}{j}=\coATIME{D}{j}$ without loss of generality. We will do
  the proof by induction over $k$, for $k=j$, we want to prove that
  $\ATIME{C}{j}=\coATIME{C}{j}$.  We will only prove $\subseteq$
  because $\supseteq$ will be true by symmetry. It is only a padding
  argument, let $f\in C$ and $L$ decided by a TM $M\in\ATIME{f}{j}$,
  then $L'=\{x1^{f(|x|)}|x\in L\}$. $L'$ can be decided by a TM $M'$
  in $\ATIME{O(n)}{j}$ hence in $\coATIME{g}{j}$ for some $g\in D$,
  then $L$ can be decided by a TM $M''\in\coATIME{f+g\circ f}{j}$
  which writes $f(n)$ ``1'' on his working tape and simulates $M'$. By
  our assumption on $C$ we then have that $M''\in\coATIME{C}{j}$,
  hence $\ATIME{C}{j}\subseteq\coATIME{C}{j}$.

  Now, let $k>j$ and suppose that the property is true for $k-1$, that
  is that $\ATIME{C}{k-1}=\coATIME{C}{k-1}=\ATIME{C}{j}$ and let $L$
  be a language accepted by a TM $M\in \ATIME{f}{k}$ with $f\in
  C$. Then on input $x$ of size $n$, we may assume without loss of
  generality that $M$ makes $f(n)$ existential steps writing $O(f(n))$
  symbols on the tape, and then make $k-j-1$ alternations. Let us say
  that this first part is done by a TM $M_{1}$. Then $M$ make a second
  part in $\Theta_{j}\TIME(f(n))$ where $\Theta$ is $\Pi$ or $\Sigma$
  depending on the parity of $k-j$, let us call $M_{2}$ the TM that
  ends the computation of $M$, since it's input tape is of size
  $O(f(n))$, $M_{2}\in \Theta_{j}\TIME{(O(n))}\subseteq
  \Theta_{j}\TIME{(D)}$. There is some $g\in D$ such that there is a
  TM $M'_{2}\in\ol \Theta_{j}\TIME(g)$ (where $\ol \Pi =\Sigma$ and
  $\ol\Sigma=\Pi$) equivalent to $M_{2}$, now we create a TM $M'$
  which begin by simulating $M_{1}$ and then $M'_{2}$; we indeed have
  only $k-1$ alternations, and the time of the computation is
  $f+(g\circ f)$ which is in $C$ by our assumptions, hence $L$ is also
  accepted by $M'\in\ATIME{C}{k-1}=\ATIME{C}{j}$ where the last
  equality is by the induction hypothesis. We obtain the result
  $\coATIME{C}{k}=\coATIME{C}{j}$ by symmetry.
\paragraph{}
The result about space is a corollary of theorem \ref{atimehigher}
when we take $G=D$ and $F=C$. We always take $g_{f}(n)=n,h_{f}=f$ and
for any $g'\in G$ $g'(h(n)+n)=g'(f(n)+n)\in C$ by the closure
assumption.
% Suppose that $\ATIME{D}{j}=\SPACE(D)$, let $L$ a language decided by
% a TM $M\in\SPACE(f)$ with $f\in C$. Let $L'=\{x1^{f(|x|)}|x \in
% L\}$, then $L'$ can be decided in $SPACE(O(n))$ by a TM verifying
% that the number of ``$1$'' is $f(x)$, and simulating $C$. Hence $L'$
% can be decided by a TM $M'\in \ATIME{g}{j}$ for some $g\in D$, hence
% $L$ can be decided by a TM $M''\in\ATIME(g\circ
% f,j)\subseteq\ATIME{C}{j}$ by witting $f(x)$ ``1'' on the tape, and
% then simulating $M'$.
\end{proof}
\begin{corollary}\label{cor:collapse} 
  If the polynomial (resp. linear) hierarchy collapses to the $j$th
  level then every exponential hierarchy collapses to the $j$th
  level. If $\PSPACE\subseteq\Sigma^{p}_{j}$ then
  $\SPACE(\ex{r})\subseteq\ATIME{\ex r}{j}$.
\end{corollary}
\begin{proof}We apply theorem \ref{little_collapse} with $D=n^{O(1)}$
  (resp. $D=O(n)$) and $C=\ex r$ for $r\ge 0$. It is easy to see that
  $C$ is closed under $D$ and under addition.
\end{proof}
\subsection{Classes of formulae}
Now we will give results for the formula formalism, there may
not be corollary of the results over general classes of formulae
because of the order of the vocabularies of the formulae, something
which does not have any exact translation in the TM setting.

\begin{lemma}\label{lem:collapse}For $j\ge0$:
  \begin{enumerate}
  \item if $\hoc{r}{j}{i}=\hoc{r}{j+1}{i}$ with $0\le i\le r+1$ then
    $\hoc{r}{j}{i}=\cohoc{r}{j}{i}=\hoc{r}{j+1}{i}=\cohoc{r}{j+1}{i}$.
  \item If $\hoc{r}{j}{i}=\cohoc{r}{j}{i}$ with $i=r$(resp. $i=r+1$)
    then
    $\hoc{r}{j}{k}=\hoc{r}{j+1}{k}=\cohoc{r}{j}{k}=\cohoc{r}{j+1}{k}$
    for all $k<r$ (resp. $k=r+1$).
    % the \eh r collapse to the $j$ level for structure of any order
    % at most $r-1$ (resp $r$). Formally, $\forall r\ge 2,j\ge
    % 1(\hob{r}{j}=\cohob{r}{j}\Rightarrow\forall k\ge j\
    % \hob{r}{j}=\hob{r}{k})$
  \item If $\hoc{r}{j}{i}=\hoc{r}{j+1}{i}$ or
    $\hoc{r}{j}{i}=\cohoc{r}{j}{i}$ for $i=r$ (resp $i=r+1$) then
    \hod rk collapses to the $j+1$th level for $k\le r$ (resp $k=r+1$).
  \end{enumerate}
\end{lemma}

The proofs are almost identical to the one for the polynomial
hierarchy which is the special case $r=i=2$.

\begin{proof}(of the lemmas) 
  For the first point, if $\hoc{r}{j}{i}=\hoc{r}{j+1}{i}$ then their
  complements are also equals, so we have
  $\cohoc{r}{j}{i}=\cohoc{r}{j+1}{i}$ hence
  $\cohoc{r}{j}{i}\subseteq\hoc{r}{j+1}{i}=\hoc{r}{j}{i}\subseteq\cohoc{r}{j+1}{i}=\cohoc{r}{j}{i}$.

  For the second point, let $\phi\in\hoc{r}{j+1}{i}$ with $i\le r$
  (resp. $i=r+1$), then $\phi=\exists\olmc X_{0}^{r}.\psi$ where
  $\psi\in\cohoc{r}{j}{r}$ (resp. \cohoc{r}{j}{r+1}), then there
  exists $\psi'\in\hoc{r}{j}{r}$ (resp. \hoc{r}{c}{r+1}) equivalent to
  $\psi$, then $\phi'=\exists \olmc X^{r}_{0}.\psi'$ is equivalent to
  $\phi$ and is in $\hoc{r}{j}{i}$, hence
  $\hoc{r}{j+1}{i}\subseteq\hoc{r}{j}{i}$.  By symmetry we also have
  $\cohoc{r}{j+1}{i}\subseteq\cohoc{r}{j}{i}$. The other side,
  $\supseteq$, is trivial, and by transitivity
  $\hoc{r}{j+1}{i}=\hoc{r}{j}{i}=\cohoc{r}{j}{i}=\cohoc{r}{j+1}{i}$.

  For the third point, by the first point of the lemma the first
  condition implies the second one, hence we are only going to use
  this condition, that $\hoc{r}{j}{i}=\cohoc{r}{j}{i}$ for $i=r$
  (resp $i=r+1$). By induction over $l\ge j$, we will prove that
  $\hoc{r}{j}{i}=\hoc{r}{l}{i}=\cohoc{r}{l}{i}$ for $i\le r$
  (resp. $i=r+1$). For $l=j$ this is the second point of the lemma, so
  assume that $l>j$ and that the property is true for $l-1$, by the
  second part of the lemma we have 
  $\hoc{r}{l}{i}=\cohoc{r}{l}{i}=\hoc{r}{l-1}{i}=\hoc{r}{j}{i}$, and
  the last equality is true by induction.
\end{proof}

What is surprising is that it seems that we do not have a proof that if
$\hoc{r}{j}{2}=\hoc{r}{j+1}{2}$ then \hod{r}{2} collapses to level
$j$. This is because, if $\phi\in\hoc r{j+1}2$, then $\phi=\exists \mc
X^{a,r}.\psi$ with $\phi\in\hoc r{j}r$ and not in \hoc r{j}2; and we
have no hypothesis about this class.  Lemma \ref{lem:collapse} is
almost what we would have wanted, but in the lemma we must bootstrap
the property with an assumption over formulae with a free variable of
order $r-1$, and in the theorem with a formula whose highest
free-variable is of degree 2. This is the descriptive complexity
translation of the question raised in \cite{sky}: if two levels of the
\eh r are equal, does the \eh r collapse? The proofs used for the
polynomial hierarchy do not work because exponentials are not closed
under composition.

\section{Variable order}\label{vo}
Variable order (\VO{}) is an extension of high-order where the orders
of the relations are not fixed any more but are variable. It was defined
in \cite{lauri}, and it was proved there that it is ``complete''; and in
fact more expressive than Turing machines, because it can decide the
halting problem, and hence also its complement.

One problem with \VO{} is that two $\alpha$-equivalent formula are not
always equivalent. \\$\forall i\forall \mc X^{i}\forall j\exists \mc
Y^{j}.(\mc X^{i}=\mc Y^{j})$ is false while $\forall i\forall \mc
X^{i}\forall i\exists \mc Y^{i}.(\mc X^{i}=\mc Y^{i})$ is true.

In this section we first give a new definition of
``\emphdex{Variable order}'' logic, equivalent to that of
\cite{lauri}, but that we consider easier to use, at least because it
respects the equivalence of $\alpha$-equivalent formulae.  Then we
prove that \VO{} contains the analytical hierarchy.

\subsection{A new definition}

\begin{definition}[Sequence of relations]
  A sequence of relations (of arity $a$) is such that the relation
  number $r$ of the sequence is of arity $a$ and order $r$.

  We will write $\mathcal X^{a}=(\mc X^{a,r})_{r \in \mathbb N^+}$ to
  mean ``$\mc X$ is a sequence of arity $a$''.
\end{definition}
\begin{definition}[Variable-order (\VO{})]Now the vocabularies will be
  over two sorts, the positive integers and the sequence of relations.
  The quantifiers of our logic will be over one of those two
    sorts.

  A variable-order formula $\phi$ is defined recursively as usual,
  such that if $\psi$ and $\psi'$ are formulae then $\psi\land\psi',
  \psi\lor \psi', \neg\psi, \forall \mathcal X^{a,r}. \psi, \exists
  \mathcal X^{a,r}. \psi, \forall r. \psi$ and $\exists r. \psi$ are also
  formula; where $\mc X^a$ are sequences of relations and $r$ is an
  order variable taking values in $\mathbb N^+$

  Finally $\mathcal X^r(\mathcal Y_1,\dots,\mathcal Y_a)$, $\mathcal
  Y=_r\mathcal X$, $r=p$ and $r<p$ are the atomic formulae where $r$
  and $p$ are variable orders and $\mc X$ and the $\mc Y_i$ are
  untyped relation variables.

  The closed formulae are defined as usual.
\end{definition}

\begin{definition}[Semantics of \VO{}]
  We will write $\mc X^{r}$ to speak of the element of order $r$ of
  the sequence $\mc X^a=(\mc X^{r})_{r\in \mathbb N}$.  $\land, \lor$
  and $ \neg$ have their usual meaning.
  \begin{itemize}
  \item $\mf A\models r=p$ if and only if $\mf A[r]=\mf A[p]$
  \item $\mf A\models r<p$ if and only if $\mf A[r]<\mf A[p]$
  \item $\mf A\models \mathcal X^{r}(\mathcal X_1,\dots,\mathcal X_a)$
    if and only if $\mf A[r]>1$ and $(\mf A[\mc X_1]^{{\mf A[r]}-1},\dots,\mf
    A[\mc X_a]^{{\mf A[r]}-1})\in\mf A[\mc X]^{{\mf A[r]}}$
  \item $\mf A\models \mc X=_r\mc X$ if and only if $\mf A[\mc \mc
    X]^{{\mf A[r]}}=\mf A[\mc \mc Y]^{{\mf A[r]}}$
  \item $\mf A\models \forall \mc X^a.\phi$ (resp.  $\mf A\models
    \exists \mc X^a.\phi$) if and only if for all sequences (resp. if
    and only if there exists one sequence) $\mathcal R^{a}=(\mc
    R^{a,r})_{r \in \mathbb N^+}$ of $a$-ary relation of every
    positive order: $\mf A[\mc X/\mc R]\models \phi$
  \item $\mf A\models \forall r.\phi$ (resp.  $\mf A\models \exists
    r.\phi$) if and only if for all (resp. if and only if there exist
    one) $i \in \mathbb N^+$: $\mf A[r/i]\models \phi$

  \end{itemize}
\end{definition}
We are now going to define $\VO'$, which is the ``variable order'' as
defined in \cite{lauri} and prove that our definition is equivalent to
theirs.

\begin{definition}
  We have an infinite number of order variables
  $r_{1},\dots,r_{n},\dots$, of first order variables
  $x_{1},\dots,x_{n},\dots$, and of untyped relation variables $\mc
  X_{1},\dots,\mc X_{n},\dots$. As in $\VO$, there are quantification
  over order variables, but there is also quantification over
  first-order variables, and the quantification over relations
  ``associates'' with it a non-free order variable. The atomic formula
  are then $\mc X^{r_{i}}(\mc Y_{1}^{i_{1}},\dots,\mc Y_{a}^{i_{r}})$
  where the exponent is associated with the relation variable, and the
  value of the exponent variable is the value of this variable in the
  scope of this formula.

  We emphasize that the value of an order variable associated
  with an untyped relation variable can change between the association
  and the atomic formula if the variable is quantified again.
\end{definition}
\begin{theorem}\label{vovo}
  \VO{} is equivalent to $\VO'$.
\end{theorem}
\begin{proof}Every formula in $\VO'$ is also a formula in \VO{} and
  its semantics is the same, so \VO{} is at least as expressive that
  $\VO'$.

  Let $\phi$ be an \VO{} formula over the vocabulary $\sigma$, let
  $\mc X_1,\dots,\mc X_n$ be the variables of $\phi$ and let
  $\sigma'=\{i_1,\dots, i_n\}$ be $n$ new order variable. We will
  create an $\VO'$ formula $\phi'$ such that $\forall
  i_1,\dots,i_n.\phi'$ is equivalent to $\phi$.

  $\phi'$ is $\phi$ where the $Q\mc X_j$ are replaced by $Q\mc
  X_j^{i_j}$ and the atomic formulae containing $\mc X_j^k$ will be
  replaced this way:
  \begin{itemize}
  \item $\mc X_j=_r\mc X_k$ is replaced by ``$\exists
    i_j,i_k.(r=i_j\land r=i_k\land \mc X_j^{i_j}=\mc X_k^{i_k})$''
  \item{}$\mc X_j^r(\mathcal X_{k_1},\dots,\mathcal X_{k_a})$ is
    replaced by ``$\exists
    i_j,i_{k_1},\dots,i_{k_a}.(r=i_j\bigwedge_{1\le b\le a}
    (r-1)=i_{k_b}\land \mc X_j^{i_j}(\mathcal
    X_{k_1}^{i_{k_1}},\dots,\mathcal X_{k_a}^{i_{k_a}}))$'' where
    ``$r-1=x$'' is a syntactic sugar for ``$x<r\land\forall o'. (
    \neg(x<o'\land o'<r))$''.
  \item{}$\mc X_j^r\in \mc X_k^p$ is replaced by ``$\exists
    i_j,i_k.(i_j=r\land i_k=p\land X_j^{i_j}=X_k^{i_p})$''.

  \end{itemize}

  \begin{lemma}Let $\sigma$ be a vocabulary, $\phi$ a formula over
    $\sigma$ such that there are $n$ relation variables, $\sigma'$ a
    set distinct of $\sigma$ of cardinality $n$, $\mf A$ a
    $\sigma$-structure and $\mf A'$ an extension of $\mf A$ over
    vocabulary $\sigma\cup\sigma'$. Then $\mf A\models
    \phi\Leftrightarrow \mf A'\models\phi'$.
\end{lemma}
This lemma implies that $\psi=\forall i_1,\dots,i_n.\phi'$ will be such
that $\mf A\models \phi\Leftrightarrow \mf A'\models\psi$.

\begin{proof}Of the lemma
  
  The proof for $\land,\lor$ and $\neg$ is an easy induction.
  \begin{itemize}
  \item If $\phi$ is $\forall \mc X_j.\psi$, then $\phi'=\forall \mc
    X_j^{i_j}$.  Then $\mf A\models \phi\Leftrightarrow \mf
    A'\models\phi'$ if and only if for all sequences of relations $\mc
    R$, $\mf A[\mc X/\mc R]\models \psi\Leftrightarrow \mf A'[\mc
    X/\mc R]\models\psi'$, and since $A'[\mc X/\mc R]$ is a
    $\sigma\cup\{\mc X\}\cup \sigma'$-structure which is an extension
    of the $\sigma\cup\{\mc X\}$-structure $\mc A[\mc X/\mc R]$ by
    induction we indeed have $\mf A[\mc X/\mc R]\models
    \psi\Leftrightarrow \mf A'[\mc X/\mc R]\models\psi'$.
  \item If $\phi$ is $\exists \mc X_j.\psi$ the proof by induction is the
    same.
  \item If $\phi$ is $\forall j.\psi$, then $\phi'=\forall
    j.\psi'$. Then $\mf A\models \phi\Leftrightarrow \mf
    A'\models\phi'$ if and only if for all positive integer $r$, $\mf
    A[i/r]\models \psi\Leftrightarrow \mf A'[i/r]\models\psi'$, and
    since $A'[i/r]$ is a $\sigma\cup\{i\}\cup \sigma'$-structure which
    is an extension of the $\sigma\cup\{i\}$-structure $\mc A[i/r]$,
    then by induction $\mc A$ we indeed have $\mf A[i/r]\models
    \psi\Leftrightarrow \mf A'[i/r]\models\psi'$.
  \item If $\phi$ is $\exists j.\psi$ the proof by induction is the
    same.
  \item If $\phi$ is $\mc X=_r\mc Y$ then $\phi'=\exists
    i_j,i_k.(r=i_j\land r=i_k\land \mc X_j^{i_j}=\mc X_k^{i_k})$. We
    will show $\mf A\models \phi\Leftrightarrow \mf A'\models\phi'$ by
    two implication.

    $\Rightarrow:$ by definition $\mf A\models \phi$ means that $\mf
    A[\mc X][r]=\mf A[\mc Y][r]$, so $r$ is a correct value for both
    $i_j$ and $i_k$ such that $r=i_j\land r=i_k\land \mc X_j^{i_j}=\mc
    X_k^{i_k}$, hence $\mf A\models\exists i_j,i_k(r=i_j\land
    r=i_k\land \mc X_j^{i_j}=\mc X_k^{i_k})$ is true.

    $\Leftarrow:$ it is clear that if $\mf A\models\exists
    i_j,i_k.(r=i_j\land r=i_k\land \mc X_j^{i_j}=\mc X_k^{i_k})$ is
    true, then $\mf A\models\mc X_j^{i_j}=\mc X_k^{i_k}$ must be true
    when $i_j=i_k=r$, so $\mf A\models \mc X=_r\mc Y$.

    The important point in this case is that the value of $i_k$ in
    $\mf A'$ has no importance.
  \item If $\phi=\mc X_j^r(\mathcal Y_{k_1},\dots,\mathcal Y_{k_a})$
    or $\phi=\mc X_j^r\in \mc X_k^p$ then $\phi'=\exists
    i_j,i_{k_1},\dots,i_{k_a}.(r=i_j\bigwedge_{1\le b\le a}
    (r-1)=i_{k_b}\land \mc X_j^{i_j}(\mathcal
    Y_{k_1}^{i_{k_1}},\dots,\mathcal Y_{k_a}^{i_{k_a}}))$ and a
    similar proof can be done, showing that the equality in $\phi'$
    will make that the value in $\mf A'$ has no importance, and will
    end the proof.
  \end{itemize}
\end{proof}
\end{proof}
There is in fact one last difficulty not treated in this proof, \VO{}
accepts that the variable order can be free and that its value can be
given in the vocabulary, which is forbidden in $\VO'$. For inductive
proofs it is easier to just consider that we can have order variables
in the vocabulary. And even if we reject the free order variable in
the formulae, we will see in section \ref{fvo} how to encode them with
relational variables in \VO.  

\subsection{Arithmetic on order variables}\label{iter-rel}
Let $r$ and $p$ be order variable, we will show that we can define
both $r+p$ and $r\times p$. In this definition we will assume that
there is at least 2 elements in the universe.

\begin{notation}
  In this section $\mc X_{a}$ will means that the variable $\mc X$ is
  of arity $a$.
\end{notation}
We cannot write $a$ as an exponent since exponent are used for order
variables. But since in the proofs we will not use list of variable
there will be non confusion.

Also in this section ``$\mc A$ contains $\mc B$'' means that $\mc B\in
\mc A$. We will use many straightforward syntactic sugar:
\begin{eqnarray}
  \label{eq:sugar}
  "i+c=j"\ed \ifte c=0\thent i=j\elset\nonumber\\\exists k>i((k+(c-1)=j)\land \neg\exists l(i<l<k))\label{plusconst}\\
  "\mc{X\cup Y=Z}"\ed\forall r,p,\mc A((\mc A^r\in\mc X^p\land \mc A^r\in\mc Y^p)\nonumber\\
  \Leftrightarrow (\mc A^r\in\mc Z^p))\\
  "\mc{A(B(C))}"\ed \mc {A(B)\land B(C)}
\end{eqnarray}
In equation \ref{plusconst} $c$ is a constant. In \ref{iter-rel} the
relation can also be <, $\in,\land,\lor$ or =.

\sss{Variable order as input}\label{fvo}
We will first need to be able to take number as input, and create a
formula $\phi_i$ such that the number of variable satisfying a monadic
second order predicate $P_i$ is equal to the the order variable
$r_i$. Formally we want that $r_i=|\{y\in A|P_i(y)\}|$ is the only
value such that $\phi(P_i,r_i)$ is true. We will not use a binary
encoding but this unary one for clarity; since we intend to prove
calculability results and not complexity one, there is no difference.

The idea we will use is to create a class of binary high-order
relation; let us call this class ``\emphdex{unique}''.

\begin{definition}
  The binary relation $\mc X$ is \emph{unique up to level $r+1$} if every
  element of the sequence $\mc X$ of order at most $r$ contains only
  one relation, which is the precedent element of the sequence
  repeated twice, and the elements of order greater than $r$ are
  empty. This imply that $\mc X^r$ contains exactly $r$ elements.
\begin{eqnarray}
  \texttt{unique}(\mc X_{2},r)\ed\forall i(1\le i\le r\Rightarrow(\mc X^{i+1}(\mc X^{i},\mc X^{i})\land \nonumber\\\forall \mc Y_{2},\mc Z_{2}\mc X^{i+1}(\mc Y^{i},\mc Z^{i} )\Rightarrow \mc X=_i\mc Y=_i\mc Z))
\end{eqnarray}
\end{definition}
We will then state that there is a bijection between the elements of
$\mc X$ and the variable $y$ that respect some property $P_i(y)$, this
will create the wanted relation between the order (of $\mc X$) and the
elements satisfying $P_i$.

A bijection will be a set $\mc T^t$ of couple of elements $\mc U^u$
(with $u=t-1$), one of the element of the couple will be an element of
$\mc X$ and the other one will be an $y$ such that $P_i(y)$. By
definition of $\mc X$, if $\mc X^r\in \mc U^u$ then for all $p<r$,
$\mc X^p\in \mc U^u$, hence we will use a more precise definition; we
will say that $X^r$ is an element of $\mc U^u$ if $r$ is the biggest
order $p$ such that $X^p\in U^u$.
 \begin{eqnarray*}\label{eq:element}
   \texttt{element}(\mc X_{2},\mc U_{2},r,u)\ed \mc X^r\in \mc U^u\land \mc X^{r+1}\notin \mc U^u\\
 \end{eqnarray*}

 It is easy to obtain such an element, we define a list $\mc E^{p}$
 this way; $\mc E^{u}=\mc U^{u}$, $\mc E^{r}=\mc X^{r}$, and for every
 $r<p<u$ $(\mc E^{p-1},\emptyset^{p-1})$ is the only relation of $\mc
 E^{p}$ where $\emptyset$ is the ``false'' relation. It is then clear
 that $\mc X^{r}\in \mc U^{u}$ and that $\mc X^{r+1}\notin \mc U^{u}$.

 Of course, every element $\mc U^u$ of the set $\mc T$ will contain
 at most two elements, one element $\mc X$ and a $y$ verifying
 $P_i$. $\mc U$ can contains also one element if $y=\mc X$. The fact
 that there are exactly one elements satisfying $\phi$ in $\mc U^{u}$
 can be called ``surjection''. %having the same image in $\mc T$
\begin{eqnarray*}
  \texttt{surjection}(\mc U,u,\phi)\ed \exists i\in \mc U(\phi(i)\land\forall j\in \mc U(\phi(j)\Rightarrow i=j))
\end{eqnarray*}

And we must also check that every element of $\mc X$ and every $y$
such that $P_i(y)$ is an element of $\mc U^u$ is contained in an $\mc
U^u$ of $\mc T^t$ . It is here that it is important that $\mc X$
contains at most one element at each level, this way we are sure that
there is exactly one element of first order in $\mc X$, if this
element is an element of $P_i$ then we will assume it is in bijection
with itself; and there is no other element of $\mc X$ that could imply
that a first order element $z$, which verify $P_i$ is in $\mc
U^{u}$. The fact that every element has got an image in $\mc T^{t}$
can be called the ``injection''.
\begin{eqnarray}
  \label{eq:relat-total}
  \texttt{injection}(\mc T,t,\phi)\ed \forall y,\mc Y(\phi(\mc Y^y)\Rightarrow\exists \mc U(\mc T^t(\mc U)\land\mc Y^{y}\in \mc U^{t-1}\nonumber \\ \land \forall \mc V(\mc T^t(\mc V)\land \mc Y^{y}\in  \mc V^{t-1}\Rightarrow \mc{U=V})))
\end{eqnarray}
Defining the bijection is just the conjunctions of injection and
surjection.
\begin{eqnarray*}
  \label{eq:bij}
  \texttt{bijection}(\mc T,t,\phi)=\forall \mc U(\mc T^{t}(\mc U,\emptyset)\Rightarrow \texttt{surjection}(\mc U,t-1,\phi))\land\texttt{injection}(\mc T,t,\phi).
\end{eqnarray*}

Assuming that there is at least one $y$ verifying $P_i$ we can tell
that there are $i$ elements $y$ verifying $P_i$, with this formula.
\begin{eqnarray}
  \label{eq:equal+}
  \texttt{equal}^+(i,P_i)\ed \exists \mc{X,T},t(\texttt{unique}(\mc X,i)\land\texttt{bijection}(\mc T,t,P_i)\land\nonumber\\
  \texttt{bijection}(\mc T,t,\lambda i.\texttt{elements}(\mc X,\mc T,i,t-1)))
\end{eqnarray}
Here $\lambda i.\phi(i)$ means that $i$ is going to be the free
variable of the property used in the formula of ``bijection''.

\paragraph{}The problem here was that there is no relation of order 0,
we are then going to encode them. We will do it this way: (1,1) means
0, ($n$,2) means $n$ and ($n$,$m$) for $m>2$ or ($m=1$ and $n>1$)
means nothing.
\begin{eqnarray*}
  \label{eq:equal}
  \texttt{equal}(i,i',\phi)\ed\ifte \neg\exists \mc X \phi(\mc X)\thent i=i'=1 \elset \texttt{equal}^+(i,\phi)\land i'=2
\end{eqnarray*}

\begin{theorem}
  \VO{} is not more expressive if the formula can have free degree
  variable .
\end{theorem}
\begin{proof}
  Let $n$ be an integer, $\sigma'=\{v_1,\dots,v_n,v'_1,\dots,v'_n\}$
  and $\sigma''=\{P_1,\dots,P_n\}$ be sets of $n$ order variables and
  monadic second order relations, let $A$ be a finite universe, let
  $\sigma$ be a vocabulary distinct from $\sigma'$ and $\sigma''$, let
  $\mf A$ be a $\sigma\cup\sigma'$-structure and let $\mf A'$ the
  $\sigma\cup\sigma''$-structure such that for $P\in\sigma$ $\mf
  A[P]=\mf A'[P]$, and for $v_i\in\sigma'$ we have $(\mf A[v_i],\mf
  A[v'_i])=|\{y\in A|u\in\mf A'[P_i]\}|$, let $\phi$ be a formula over
  vocabulary $\sigma\cup\sigma'$. Then $\mf A\models
  \phi\Leftrightarrow\mf A'\models \forall_{1\le i\le
    n'}v_i,v_i'(\phi\bigwedge_{1\le i\le n}\texttt{equal}(v_i,v_i',P_i))$.
\end{proof}
\sss{Addition}
We now want to be able to add order variables. The idea will be the
same, $r+p=q$ if there is a bijection between a relation of order $q$
and the union of a relation of order $r$ and a relation of order $p$.
We will do it by having $\mc Y^p$ be in $\mc Z^q$, and quantify a
bijection between elements of $\mc X^r$ and the elements of $\mc Z^q$
of order higher than $p$.
\begin{eqnarray}
  \label{eq:plus+}
  \plu^+(r,p,q)\ed\exists \mc{T,X,Y,Z},t,(\texttt{unique}(\mc X,r)\land\texttt{unique}(\mc Y,p)\land\nonumber\\
  \texttt{unique}(\mc Z,q)\land\mc Y^r\in\mc Z^q\land \texttt{bijection}(\mc T,t,\lambda i.\texttt{elements}(\mc X,\mc U,i,t-1))\nonumber\\
  \land \texttt{bijection}(\mc T,t,\lambda i.\texttt{elements}(\mc Z,\mc U,i,t-1)\land i>p)\land \texttt{different}(\mc X,\mc Z))
\end{eqnarray}
We need to make sure that the bijection between elements of $\mc Z$
and the one of $\mc X$ is correct by checking that there is no element
that are both in $\mc X$ and $\mc Z$ , this is the point of
different$\mc{(X,Z)}$.
\begin{eqnarray}
\label{eq:plus}
\texttt{different}(\mc X,\mc Z)\ed\forall i\mc X\not=_i\mc Z% \\
\end{eqnarray}

Finally, using the code for 0 and positive integers of the last
subsection, we can define the addition of $\mathbb N$.
\begin{eqnarray}\plu(r,o',p,p',q,q')\ed \ifte o'=1 \thent(p=q\land p'=q')\elset\nonumber\\
  (\ifte p'=1\thent (r=q\land o'=q')\elset(\plu^+(r,p,q)\land q'=2))
\end{eqnarray}
\sss{Multiplication}
Finally we want to code the multiplication of order, once again the
formula $r\times p=q$ will choose relations $\mc {X,Y}$ and $\mc Z$,
unique up to order $r,p$ and $q$ respectively, such that there is a
bijection between the elements of $\mc Z$ and the Cartesian product of
the elements of $\mc X$ and of the elements of $\mc Y$.
\begin{eqnarray}
  \label{eq:time+}
  \texttt{times}^+(r,p,q)\ed\exists \mc{T,X,Y,Z},t,(\texttt{unique}(\mc X,r)\land\nonumber\\
  \texttt{bijection}(\mc T,t,\lambda i,j.\texttt{elements}(\mc X,\mc U,i,t-1)\land\texttt{elements}(\mc Y,\mc U,j,t-1))\nonumber\\
  \land\texttt{unique}(\mc Y,p)\land\texttt{unique}(\mc Z,q))
\end{eqnarray}
Of course we now can extend the multiplication over every non negative
integers.
\begin{eqnarray}\label{eq:times+}\texttt{times}(r,o',p,p',q,q')\ed \ifte (o'=1\lor p'=1) \thent(q'=q=1)\nonumber\\
  \elset(\texttt{times}^+(r,p,q)\land q'=2)
\end{eqnarray}

\sss{Set of natural numbers}
We can define any set $S\subseteq \mathbb N$ in \VO{} as a sequence of
relation $\mc X^{1}$ such that if $i-1\in S$ then $\mc X_{i}=\top$ else
$\mc X_{i}=\bot$. We can of course assert that $\mc X$ is a correct
code with
\begin{eqnarray*}
  \texttt{correct-set}(\mc X^{1})\ed\forall i( X=^{i}\top^{i} \lor
  X=^{i}\bot^{i})
\end{eqnarray*}
and that $n\in \mc X$ with 

\begin{eqnarray*}
  \texttt{in}(n,\mc X)\ed \mc X=^{n+1}\top.
\end{eqnarray*}

\subsection{VO contains the analytical hierarchy}
%\subsection{VO is equal to the arithmetical hierarchy}
% It is proved in \cite{lauri} that \VO{} can express the
% halting-problem, and since \VO{} is close by negation, it also
% express its negation; hence \VO{} contains at least in
% $\Sigma^0_1\cup\Pi_1^0=\Delta^0_2$. We will show that VO contains
% the arithmetical hierarchy.
% \begin{definition}[Bounded quantifier ]
%   A bounded quantifier is a quantifier of the form $\phi=Q x<y(\psi)$
%   where $Q$ is a quantifier and $y$ is free in $\phi$.

%   A formula is a bounded formula if and only if every quantifiers are
%   bounded.
% \end{definition}
% \begin{definition}[Arithmetical hierarchy]
%   Let $\sigma=\{+,\times,=,c_1,\dots,c_n\}$ where the $c_i$ are
%   constants. Let $\mf N$ be a $\sigma$-structure over the universe
%   $\mathbb N$ such that every arithmetical operations has their usual
%   meaning.

%   Then let $\Sigma_0=\Pi_0$ be the set of bounded formula, $\phi$ is
%   in $\Sigma_{i+1}$ if it is in the form $\phi=\exists \ol
%   x\psi$ where $\psi$ is in $\Pi_{i}$. A formula is in $\Pi_i$ if it
%   is the negation of a formula in $\Sigma_i$. Let
%   $\Delta_i=\Sigma_i\cap\Pi_i$, $\Delta_i$ is the $i$th level of the
%   arithmetical hierarchy.

%   The arithmetical hierarchy (\AH) is equal to the union of the
%   $\Delta_i$; $\AH=\bigcup_{i\in\mathbb N}\Delta_i$.
% \end{definition}
\begin{definition}[Analytical hierarchy(\AnH)]
  Let $\sigma=\{+,\times,=,c_1,\dots,c_n,S_{1},\dots, S_{m}\}$ where
  the $c_i$ are constant natural numbers and the $S_{i}$ are constant
  sets of natural numbers. Let $\mf N$ be a $\sigma$-structure over the
  universe $\mathbb N$ such that every arithmetical operation has
  its usual meaning.

  Then let $\Sigma^{1}_0=\Pi^{1}_0=\Delta_{0}^{1}$ be the set of
  formula with quantification only on first order variables. The
  formula $\phi$ is in $\Sigma^{1}_{i+1}$ if it is in the form
  $\phi=\exists \ol X\psi$ where $\psi$ is in $\Pi_{i}$, ``$\exists
  \ol X$'' is a quantification over the subset of $\mathbb N$.  A
  formula is in $\Pi_i$ if it is the negation of a formula in
  $\Sigma^{1}_i$. Let $\Delta^{1}_i=\Sigma^{1}_i\cap\Pi^{1}_i$,
  $\Delta^{1}_i$ is the $i$th level of the analytical hierarchy.

  The analytical hierarchy (\AnH) is equal to the union of the
  $\Delta^{1}_i$; $\AH=\bigcup_{i\in\mathbb N}\Delta^{1}_i$.
\end{definition}
% \begin{definition}[Turing degree]
%   Let $S$ be a set of languages, and $S'$ the set of language that can
%   be recognized by Turing Machine with oracles in $S$. Then the $n$th
%   Turing degree $\emptyset^{(n)}$ is defined as $\emptyset$ if $n=0$,
%   else as $\emptyset'^{(n-1)}$.
% \end{definition}

% It is known that $\emptyset^{(n)}=\Delta_{n-1}$

\begin{theorem}
We have
$\AnH\subseteq\VO$\end{theorem}\label{anhvo}
\begin{proof}This section explained how to transform input into order
  variable, and how to add and multiply order variable; it also
  explained how to quantify sets of natural numbers, and express that
  a number is inside of the set.  Then every formula of \AnH{} can be
  easily encoded into \VO.

% \paragraph{$\supseteq$}
% Let $\phi$ be the number of order quantifier alternation, we will give
% an algorithm in $\emptyset^{(n+1)}$ to calculate $\phi$, this will
% prove that $\phi\in \Delta_n$ and hence $\phi\in\AH$.

% Without loss of generality we assume that $\phi$ is in normal form,
% beginning by $n$ sets of existentials quantifiers separated by
% negations. $\phi=\exists\ol x_1\neg\exists\ol
% x_2,\neg,\dots,\neg\ol x_n\psi$ with $\psi$ order
% quantifier-free.

% By induction over $n$, if $n=0$ then the formula is an \HO{} formula,
% and we already know they are decidable in bounded time, hence they are
% in $\emptyset^{(1)}$. If $n>0$ then we have an oracle in
% $\emptyset^{(n)}$ to decide $\phi=\exists\ol
% x_2,\neg,\dots,\neg\ol x_n\psi$, then we are going to enumerate
% every tuple of variable $\ol x_1$ and for each tuple we will
% query this oracle, we will accept when the oracle reject. This proves
% that $\phi\in\emptyset^{(n+1)}$ and hence $\AH\subseteq\VO$.

% This ends the theorem.
\end{proof}
%\section{Conclusion}
\section{Open problems}\label{open}

\paragraph{Direct equality between classes}
When many classes are equal, it may be interesting to find a way to
directly transform the formulae without needing to encode a Turing
machine. So we may want to find a direct translation from $\hoa
r(\NPFP)$ to $\hoa r(\PFP)$, $\hoa{r+1}(\TC)$ or $\hoa r(\AIFP)$. We
also would like to prove that $\hoa {r+1}(\IFP)\subseteq\hoa
r(\APFP)$.

\paragraph{\eh r}Is $\hob{r}{j}$ a strict subset of $\hob{r}{j+1}$ ?
For $r=2$ this question is: ``Does the polynomial hierarchy collapse
to the $j$th level ?''. And as we saw in theorem \ref{cor:collapse} if
we can prove that there is at least one $r$ such that the \eh r does
not collapse to the $j$th level, then the same result is true for all
$p<r$. This may eventually be a way to prove that the polynomial
hierarchy does not collapse to some level, hence that $\P\ne\NP$.

We also wonder if $\hoa{r}(\IFP)$ is strictly contained in
$\hob{r+1}{1}$, for $r=1$ it is the question $\P\ne\NP$.

More surprising, we leave as open the question: If ($\hoc rjf=\cohoc
rjf$ or $\hoc rjf=\hoc r{j+1}f$), for $r>2$, does \hod rf collapse to
the the $j$th level ? In general, for a class of function $C$ what is
the condition over $C$ such that $\ATIME{C}{j}=\coATIME{C}{j}$ or
$\ATIME{C}{j}=\ATIME{C}{j+1}$ implies that the class $\ATIME{C}{.}$
collapse to the $j$($j+1$ ?) level. We gave sufficient condition but
can not prove that they are necessary. We think that those implication
must be true, because for them to be false we must have that, for some
$j$, $j$ or $j+1$ alternation does not change the expressivity, but 
for some $k>j$, $k$ alternations is more expressive; this seems to make
no sens.

\paragraph{Relational machines}

Relational machines where introduced in \cite{AbVi}, and extended in
\cite{nfp}; they are an extension of the Turing machines with relation
register. The input are given in the register and not on the tape,
which remove the implicit order that Turing machines usually has on
the input. The machine can, as usual, write on the tapes, read the
tapes, but can also apply boolean operations to the registers and
check if a register is empty. The input is then measured as the number
of different types of elements in the input; because the size of the
input can not be known by relational machines.

It was proven that relational-\P, relational-\NP{}, relational-\PSPACE{}
and relational-\EXP{} are equivalent to $\FO(\IFP)$, $\FO(\NIFP)$,
$\FO(\PFP)$ and $\FO(\AIFP)$, and that two relational classes are
equivalent if and only if the usual classes are equivalent.

We think that it may be interesting to find a correct extension to
those relational machines to simulate high-order formulae. In
particular it may give let us transform the ``reasonable input''
assumption into something more formal over those relational machines.

\paragraph{Fixed arity high-order}
We discussed Monadic High Order, which is the special case of
``maximal-arity'' beeing 1 as defined in \cite{arity}. It may be
interesting to give a better caracterisation of expressivity of logics
in function of maximal-arity, basic-arity \cite{kolo} or other
restriction of arity.

\paragraph{Restrictions}Is there a good way to define Horn and Krom
formulae in high-order ? As stated in section \ref{horn}, finding a
correct definition with good properties seems to be not trivial. 
Finally, over high-order, is there some other syntactic restriction
which give interesting properties?

\paragraph{Games}
In first and second order logic, games, like the Ehrenfeucht-Fraïssé
(see chapter 4 of \cite{libkin}) ones, are tools to prove that some
queries are not expressible in a given logic. It would be interesting
to extend these games over the high order classes.  We might even define a game
for every class, which would let us prove that some queries are not
elementary.

Those games would be very hard to win for the duplicator, so it would
then be interesting to try to find easier games.

\paragraph{Other extensions}
What would be the effect of adding counting quantifiers, or unary
quantifiers, over high order logic? How would the different
infinitary logics be more expressive with high order? (The definition
of those logic can be found in chapter 8 of \cite{libkin}.)

\paragraph{Variable order}
What is the exact upper bound on the expressivity of variable
order? We give the analytical hierarchy as a lower bound,
$\AnH\subseteq\VO$, and we conjecture this to be an equality, but
coding a variable order formula into the analytical hierarchies seems
to be a nontrivial technical task.

What would be the expressivity of $\VO'$ if the order variables could not be
quantified many times? Since the variable should be quantified before
the formula it is associated to is quantified, it could be a severe
restriction to the expressivity of the language. The author thinks
that this would express exactly the class of functions computable in
elementary time. (This class is at least a lower bound, since this
version of $\VO'$ would be a superset of $\hoa{r}$ for any value of
$r$).

The idea behind this assumption is that with a finite number of order
of variable it is impossible to find difference between two relations
of order sufficiently high, if we can decide what is the exact bound
for a given number of order variable, let us say $b$, then we can
replace every $Q i. \phi$ by $(Q i<b)\phi$, hence the language is
decidable and it seems that this kind of formulae can be written as
formulae in $\HO$(with $b$ differents value of order from 1 to $b$ for
every relation variables).

% \subsection{Acknowledgment}
% This research was done during my  first year of ``mastère'' at École
% normale supérieur and partially suported by its computer science
% departement. It was an internship with Professor David A. Mix
% Barington at University of Massachusetts at Amherst, and I would like
\bibliography{fo}

\end{document}